\documentclass[journal]{IEEEtran}
\usepackage{ifpdf}
\ifpdf
    \usepackage[pdftex]{graphicx}
    \usepackage[update]{epstopdf}
\else
	\usepackage{graphicx}
\fi
\usepackage{caption}
\usepackage{subcaption}
\usepackage{wrapfig,bbm}
\usepackage{enumitem,color}
\usepackage{amsthm,multirow}
\newtheorem{theorem}{Theorem}
\newtheorem{lemma}{Lemma}

\newtheorem{example}{Example}
\newtheorem{prop}{\textbf{Proposition}}

\usepackage[mathscr]{euscript}
\usepackage{hyperref}
\usepackage{soul}
\usepackage{eufrak}
\usepackage{bbold}
\usepackage{array}
\usepackage{amsmath,amsthm}
\usepackage{mathrsfs}
\usepackage{amssymb}
\usepackage{amstext}
\usepackage{cite,setspace}
\usepackage{xcolor}
\usepackage{algorithm2e}
\newtheorem{definition}{\textbf{Definition}}

\DeclareGraphicsExtensions{.pdf,.png,.jpg,.eps}

\ifodd 1

\else

\fi

\ifodd 1

\else

\fi

\ifodd 1
\newcommand{\congc}[1]{{\color{red}(Cong: #1)}}
\else
\newcommand{\congc}[1]{}
\fi

\newcommand{\sumV}[1]{\sum_{\substack{\vec{v}\doteq\mathcal{A} \\ #1}}}
\def\cA{\mathcal{A}}
\def\cC{\mathcal{C}}
\newcommand{\M}[1]{\prod_{k=1}^m f_k(\vec{v}_{#1})}
\newcommand{\fprod}[1]{\prod_{#1}^m f_k(\vec{v}_{1})}
\newcommand{\fprodT}[1]{\prod_{#1}^m f_k(\vec{v}_{2})}

\begin{document}

\title{On Top-$k$ Selection from $m$-wise Partial Rankings via Borda Counting}

\author{Wenjing Chen, Ruida Zhou, Chao Tian, and Cong Shen
\thanks{A preliminary version of this work has been presented at the 2020 IEEE International Symposium on Information Theory (ISIT) \cite{chen2020isit}.}
\thanks{W. Chen, R. Zhou and C. Tian are with Department of Electrical and Computer Engineering, Texas A\&M University, TX, USA. }
\thanks{C. Shen is with the Charles L. Brown Department of Electrical and Computer Engineering, University of Virginia, VA, USA.}
\thanks{The work of C. Shen was supported in part by a Virginia Commonwealth Cyber Initiative (CCI) cybersecurity research collaboration grant.}
}

\maketitle

\begin{abstract}

We analyze the performance of the Borda counting algorithm in a non-parametric model. The algorithm needs to utilize probabilistic rankings of the items within $m$-sized subsets to accurately determine which items are the overall top-$k$ items in a total of $n$ items. The Borda counting algorithm simply counts the cumulative scores for each item from these partial ranking observations. This generalizes a previous work of a similar nature by Shah et al. using probabilistic pairwise comparison data. The performance of the Borda counting algorithm critically depends on the associated score separation $\Delta_k$ between the $k$-th item and the $(k+1)$-th item. Specifically, we show that if $\Delta_k$ is greater than certain value, then the top-$k$ items selected by the algorithm is asymptotically accurate almost surely; if $\Delta_k$ is below certain value, then the result will be inaccurate with a constant probability. In the special case of $m=2$, i.e., pairwise comparison, the resultant bound is tighter than that given by Shah et al., leading to a reduced gap between the error probability upper and lower bounds. These results are further extended to the approximate top-$k$ selection setting. Numerical experiments demonstrate the effectiveness and accuracy of the Borda counting algorithm, compared with the spectral MLE-based algorithm, particularly when the data does not necessarily follow an assumed parametric model. 
\end{abstract}

\section{Introduction}

The problem of rank aggregation has drawn considerable attention due to its diverse applications in information retrieval, recommending system, and social science; see e.g., \cite{brin1998anatomy, baltrunas2010group, caplin1991aggregation,Huin2019, AgrawaliIEEE2017}. The issue at hand usually involves a total of $n$ items, where rankings of the items in certain small subsets are observed, and a target global feature of the full ranking of the $n$ items needs to be determined. 
Initial effort focused on the problem of constructing a full ranking of the $n$ items that is consistent with the partial rankings which have been observed in a deterministic manner \cite{fagin2004comparing}. Later research attention has turned to the probabilistic setting, where instead of deterministic observations on the ranking of a subset of items, ranking results are observed following certain probability law. 

A well studied case is rank aggregation from pairwise comparisons under certain parametric models. The Bradley-Terry-Luce (BTL) model \cite{bradley1952rank} has been adopted by many existing works, e.g., \cite{wauthier2013efficient, negahban2016rank, chen2015spectral,AlsanBTL2018, ArslanIEEE2017}, where each item  has an unknown underlying score, and the probability 
distribution 
of observing one item ranked over the other is simply a
Bernoulli distribution
parametrized by the underlying scores. Under this parametric model, the rank aggregation problem essentially reduces to an estimation problem. 
Instead of pairwise comparisons, Jang et al. \cite{JangIEEE2018} studied the generalized setting of using $m$-wise probabilistic rankings under the parametric Plackett-Luce (PL) model to select the top-$k$ items \cite{plackett1975analysis}.

Algorithms based on explicit parametric models are likely to perform poorly if the data does not match the model.  It is therefore desirable to use algorithms that do not explicitly rely on any parametric model, i.e., using a non-parametric model. In a recent work \cite{shah2017simple}, Shah et al. considered a counting algorithm in such a setting, which simply keeps track of the number of times each item wins in pairwise comparisons. It was shown that the method is accurate and robust for top-$k$ and approximate top-$k$ selections.

In this work, we extend the study of the non-parametric model to investigate rank aggregation from probabilistic $m$-wise partial rankings. The natural algorithm of choice is the Borda counting method \cite{borda1784memoire}. In this method, observations of partial rankings for items within subsets of items are collected. In each (sampled) observation, the item ranked at the $i$-th position within the subset is given a score of $\beta_i$; the eventual overall ranking of the items are produced using the cumulative scores obtained by the items. This method has been used in democratic elections and sports. Two example scenarios are as follows (see e.g., \cite{fraenkel2014borda} and \cite{kaiser2019strategy}): 
\begin{itemize}
    \item In some sports, the athletes are ranked using their competition results in the events that they participate over a period of time. For example, results in different heats at a tournament for bicycle track racing can be used to select athletes into a national team, or in some sports, athlete ranking results from multiple tournaments over a full year can be collected for aggregation to generate an annual ranking. 
    \item In some election systems, the voters are asked to rank the candidates, and the eventual election results are produced using Borda counting for aggregation to find the top candidate or candidates. Two well known scoring systems are when the candidate ranked at the $i$-th position is given a score of $n-i$, sometimes called the tournament system, and when the candidate ranked at the $i$-th position is given a score of $1/i$, sometimes call the Dowdall system.
\end{itemize}

We studied the error probability of the Borda counting method, and establish an almost matching pair of the conditions in a similar manner as those given by Shah et al. \cite{shah2017simple} for the pairwise  setting. We also show that the difference between the achievable bound and the converse bound depends on $m$, as well as the choice of the scoring system in the Borda counting procedure. 

Although the main technical tools and approaches mostly follow that of Shah et al. \cite{shah2017simple}, this generalization is by no means trivial. In the analysis of the error probability, we observe that the error probability of the counting procedure can be upper-bounded differently in the high observation probability regime and low observation probability regime. More precisely, Hoeffding's inequality is more effective for the former while Bernstein's inequality is a better choice for the latter. 
The bound given by Shah et al. relies only on the Bernstein's inequality, and thus the proposed bound is a strict improvement even in the pairwise setting. The new bound also leads to a reduced gap between the upper and lower bounds of the error probability, which is particularly evident in the asymptotic regime of large $n$. These results are further extended to the approximate top-$k$ selection setting with a similar set of upper and lower bounds. Through numerical simulations, we observe that the Borda counting procedure is competitive to the spectral MLE-based algorithm, particularly in the high observation probability regime and on real-world data where parametric models may be a mistmatch.

In summary, our main contributions are: 
    \begin{itemize}
        \item We generalize the non-parametric model from pairwise comparisons to $m$-wise partial rankings. Additionally, justifications for this general non-parametric model are provided.
        \item Second, we establish an upper bound of the error probability, which is in fact tighter in the pairwise case than \cite{shah2017simple}. We also established a lower bound through a more sophisticated construction.
        \item Lastly, we conduct numerical study to compare the Borda counting algorithm with the spectral MLE based algorithm. The results show that the Borda counting algorithm is both  efficient and accurate.
    \end{itemize}

The remainder of this paper is organized as follows. In the next section, we formally introduce the problem. Section \ref{section: top-k selection} presents the upper bound and the lower bound of the error probability and the proof outlines. In Section \ref{section: hamming error} we generalize the results in Section \ref{section: top-k selection} to obtain similar bounds when top-$k$ selection is allowed within certain Hamming distance. Section \ref{section: experiment} presents experiment results to illustrate the performance of the Borda counting-based algorithm in comparison to the spectral MLE-based algorithm. Technical proofs are deferred to the appendix. 

\section{Notation and Preliminaries}
\label{sec: problem setup}

In this section, we introduce the ranking data collection setting, the Borda counting algorithm, and the associated score in the problem. 

\subsection{Probabilistic ranking collection}

There are a total of $n$ items indexed by $[n]:=\{1,2,\ldots,n\}$ in the problem setting. Let $m$ be a positive integer where $2\leq m\leq n$. Observations of partial rankings are collected as follows. For each subset of $[n]$ containing $m$ items, say $\mathcal{A}\subseteq [n]$ with $|\mathcal{A}|=m$, a total of $r$ rounds of data collections are performed. In each round, a random ranking sample of the items in $\mathcal{A}$ is generated. It is then observed with probability $p\in (0,1]$, and not observed with probability $1-p$. The ranking observations in different rounds and for different subsets of items are generated independently.
The underlying unknown probability distribution that generates the random ranking samples in the subset  $\mathcal{A}=\{v_1,v_2,\ldots,v_m\}$ is written as follows. 
\begin{align}
    M_{v_1,v_2,\ldots,v_m}:=\text{Prob}(\text{ranking follows } v_1,v_2,\ldots,v_m ). 
\end{align}
In the sequel, we denote the vector $(v_1,v_2,\ldots,v_m)$ as $\vec{v}$, and write $\vec{v}\doteq\mathcal{A}$ to indicate that $\vec{v}$ is a permutation of the items in the set $\mathcal{A}$. 
It is clear that 
\begin{align}
    \sum_{\vec{v}\doteq\mathcal{A}}  M_{\vec{v}} = 1.
\end{align}
We have a total of $n \choose m$ such probability distributions in our system, one for each subset of $\mathcal{A}$. Note that we do not assume a given generative model from the full ranking of the items to the partial rankings, i.e., these $n \choose m$ distributions are not mutually constraining. The collection of such probability distributions is denoted as $\mathcal{M}$.

\begin{example}
Consider the case $n=4$, and $m=3$. Therefore the possible sets of $\mathcal{A}$ of cardinality 3 are
\begin{align}
    \{2,3,4\},\{1,3,4\},\{1,2,4\},\{1,2,3\}.
\end{align}
Suppose $r=2$, $p=0.5$, then rankings within each of the subsets above are observed with probability $0.5$ in each round. Suppose for the subset $\mathcal{A}=\{1,3,4\}$, the probability distribution is
\begin{align}
    M_{1,3,4}=0.35, M_{1,4,3}=0.1, M_{3,1,4}=0.1, \notag\\
    M_{3,4,1}=0.1, M_{4,1,3}=0.15, M_{4,3,1}=0.2.
\end{align}
Then an observation on the ranking is most likely to be $(1,3,4)$ on this subset. Note that the distribution does not need to follow a parametrized model, and the probability values are not known a priori. 
\end{example}

\subsection{The Borda counting algorithm}

The Borda counting algorithm records the cumulative score of each item in all probabilistic partial ranking observations, where in each observation an item receives a score $\beta_i$ when it is ranked in the $i$-th position. The scoring vector $\vec{\beta}=(\beta_1,\beta_2,\ldots,\beta_m)$ is a non-decreasing sequence and we assume $1=\beta_{1}\geq \beta_{2}\geq ...\geq \beta_{m}\geq 0$. The score $\beta_i$ represents the ``weight" or ``value" of being ranked in position $i$. After $r$ rounds of ranking data collection, the items are then ranked from top to bottom according to the cumulative scores they receive; the estimate of the top-$k$ set is thus the set of items with highest cumulative scores. This set is denoted as $\tilde{\mathcal{S}}_k$.

 The ranking results clearly depend on the particular choice of the scoring systems, i.e., the $\beta$ vector. As mentioned earlier, two popular scoring systems are the tournament-style system where $\beta_{i}=m-i$ and the Dowdall systems where $\beta_i=1/i$. In the former system, obtaining the top ranking is less important than in the latter, and thus if the former scoring system is used in an election, candidates ranked with less  extreme positions 
 will be scored relatively high even if they do not stand out among the candidates. The design of the scoring systems is a frequent point of contention in election systems and in sports, and it involves other less understood factors such as human psychology, which is beyond the scope of this study. We shall assume the scoring system in our problem setting is fixed before hand.

\begin{example}
Continuing our previous example, let us suppose the ranking observations in the two rounds are
\begin{align}
    [(2,1,3), (2,3,1)], (3,2,4), (3,4,1), (1,2,4),
\end{align}
and the scoring system is $\vec{\beta}=(1,0.5,0)$. Note that ranking for the subset $\{1,2,3\}$ was observed twice, and others are only observed once. The eventual cumulative scores of the items are 
\begin{align}
    (1.5,3,2.5,0.5).
\end{align}
Therefore, the ranking of the four items at the end is $(2,3,1,4)$. The top-$2$ items would be $2,3$ using the Borda counting method.
\end{example}

To facilitate subsequent analysis, let us denote the score that item-$a$ receives in the $\ell$-th round in the competition among the items in $\mathcal{A}=a\cup \mathcal{A}^-\subseteq [n]$ as $X_{a,\mathcal{A}^-}^{(\ell)}$, where $\mathcal{A}^-\subseteq [n]$ has cardinality $m-1$. After $r$ rounds, the total score received by item-$a$ is therefore
\begin{equation}
W_a=\sum_{\ell\in{[r]}}\sum_{\mathcal{A}^-\subseteq{[n]}\setminus \{a\}}X_{a,\mathcal{A}^-}^{(\ell)}.
\end{equation}

In the remainder of this work, we use $\mathcal{A}$ to enumerate subsets of $[n]$ with cardinality $m$ without explicitly stating $m$ in the notation; similarly, $\mathcal{A}^-$ is used for those of cardinality $m-1$, and $\mathcal{A}^{--}$ for those of cardinality $m-2$.

\subsection{The associated scores}

Let us consider the observation collection process more carefully. In the probabilistic ranking among the items in $\mathcal{A}$ with $|\mathcal{A}|=m$, denote the probability that item-$a$ ranks at the $t$-th position in the set $\{a\}\cup \mathcal{A}^-$ as $R_{a,\mathcal{A}^-}(t)$, and thus
\begin{align}
R_{a,\mathcal{A}^-}(t)=\sum_{\substack{\vec{v}\doteq\{a\}\cup \mathcal{A}^-\\v_t=a}}M_{\vec{v}}.
\end{align}

Then for each item $a$, the expected cumulative score is 
\begin{align}
S_a=p\sum_{\substack{\mathcal{A}^-\subseteq [n]\setminus\{a\}}}\sum_{t=1}^{m}\beta_{t}R_{a,\mathcal{A}^-}(t).
\end{align}
Since $\beta_i\leq1$, we have
\begin{align}
S_a&\leq p\sum_{\substack{\mathcal{A}^-\subseteq [n]\setminus\{a\}}}\sum_{t=1}^{m}R_{a,\mathcal{A}^-}(t)\notag\\
&\leq p\sum_{\substack{\mathcal{A}^-\subseteq [n]\setminus\{a\}}}1\leq p\binom{n-1}{m-1}.
\end{align}
We define the associated score, i.e., the expected normalized cumulative score of an item $a$ as
\begin{align}
&\tau_a
=\frac{1}{\rho_{n,m}}\bigg{(}\sum_{\substack{\mathcal{A}^-\subseteq [n]\setminus\{a\}}}\sum_{t=1}^{m}\beta_{t}R_{a,\mathcal{A}^-}(t)\bigg{)},
\end{align}
where the normalization factor $\rho_{n,m}$ is 
\begin{align}
&\rho_{n,m}=\binom{n-1}{m-1}.
\end{align}
We will view the ``true" top-$k$ items as the set of $k$ items with the highest associate scores, denoted as $\mathcal{S}_k^*$.

Using this notation, we can write the probability distribution of $X_{a,\mathcal{A}^-}^{(\ell)}$ as
\begin{equation}
\mathbf{Pr}(X_{a,\mathcal{A}^-}^{(\ell)}=\beta)=
\begin{cases}
pR_{a,\mathcal{A}^-}(t) &\beta=\beta_{t},\, t=1,2,\ldots,m;\\
1-p&\beta=0,
\end{cases}
\end{equation}
It is then clear that 
\begin{equation}
\mathbb{E}\left(X_{a,\mathcal{A}^-}^{(\ell)}\right)=p\sum_{t=1}^{m}\beta_{t}R_{a,\mathcal{A}^-}(t),
\end{equation}
and we thus have
\begin{align}
\mathbb{E}W_a=p\rho_{n,m}\tau_a,
\end{align}
i.e., the cumulative score in the Borda counting algorithm of an item $a$ is a scaled unbiased estimator of $\tau_a$.

It should be noted that in non-parametric settings, the concept of a ``true" ranking can be ambiguous, and the total ranking based on the associated scores is simply one such choice.  Nevertheless, we do expect that if the data is indeed generated from a parametric model, the ranking based on the associated score will be consistent with the ranking generated by the model. We indeed have the following reassuring result, whose proof is given in the appendix.
\begin{lemma}
\label{lemma:consistency of para model}
If the partial rankings samples are produced from the PL model, then the ranking using the associated scores is consistent with the ranking of the weight vector, regardless of the score system  $\vec{\beta}$ being used, as long as it satisfies $1=\beta_{1}\geq \beta_{2}\geq ...\geq \beta_{m}\geq 0$.
\end{lemma}

The difference between the $k$-th highest associated score and the $(k+1)$-th highest associated score can be written as
\begin{align}
\Delta_k=\tau_{(k)}-\tau_{(k+1)}.
\end{align}
This quantity is important because if it is large, we would expect the top-$k$ items are easy to find. The same observation has been made in non-parametric models in prior works for the pairwise setting \cite{shah2017simple}. In parametric settings, a  similar definition of the $k$-th threshold was also adopted in \cite{chen2015spectral} \cite{jang2017optimal}.

\section{Top-$k$ Selection}
\label{section: top-k selection}
\subsection{Main Result}
Our main result is a pair of almost matching conditions regarding the accuracy of the Borda counting-based method. Before presenting the main result, we first consider the following set of partial ranking probabilities distributions:
\begin{align}
\mathscr{F}_k(\alpha)=\left\{M\in\mathcal{M}: \Delta_k\geq\alpha\sqrt{\frac{\log n}{rp\rho_{n,m}}}\right\}.\label{eqn:Falpha}
\end{align}
The probability distributions in the set $\mathscr{F}_k(\alpha)$ has a relatively large $\Delta_k$, which intuitively suggests that the top-$k$ items may be easy to identify. This intuition is made formal in the following theorems. 

\begin{theorem}
\label{theorem:forward2}
For any $\alpha>0$, the probability of choosing incorrect top-$k$ items using the Borda counting-based method for any items with $M\in\mathscr{F}_k(\alpha)$ is upper bounded as
\begin{align}
&\sup\limits_{M\in\mathscr{F}_k(\alpha)} \mathbb{P}_M[\tilde{\mathcal{S}}_k\neq \mathcal{S}_k^*]\nonumber\\
&\leq
\begin{cases}
  k(n-k)n^{-\frac{\alpha^2}{(4-2p)}}& 0<p\leq p_0,\\
  k(n-k)n^{-\frac{\alpha^2p}{(2(1-\beta_m)^2\frac{m-1}{n-1}+\frac{n-m}{n-1})}}& p_0<p\leq1,
\end{cases}
\end{align}
where $p_0\triangleq\min\left(\frac{1}{2},1-\sqrt{1-(1-\beta_m)^2\frac{m-1}{n-1}-\frac{1}{2}\frac{n-m}{n-1}}\right)$. 
\end{theorem}

An outline of the proof for Theorem~\ref{theorem:forward2} is provided in Section \ref{sec:proofoutline}, and we give the complete proof in Section \ref{sec:prooflemmas}. From Theorem~\ref{theorem:forward2}, we see that the bound behaves differently in the low observation probability regime ($0<p\leq p_0$) and the high observation probability regime ($p_0<p\leq1$). This is because the bounds in the two regimes are in fact derived using different concentration inequalities, which are shown to be more effective in the respective regimes.

Our next result establishes the fundamental limit of the score-based method. 
\begin{theorem}
\label{theorem:converse}
Let $n$ and $k$ be chosen with $2k\leq n$. If $\alpha\leq\bar{\alpha}(g,m,\vec{\beta})\triangleq\frac{\sqrt{2}}{7}g(n,m,\vec{\beta})\sqrt{\frac{1}{h(n,m)\rho_{n,m}}}$, $p\geq\frac{\log n}{4rh(n,m)}$, and $n\geq 7$, then the error probability of any estimator $\hat{\mathcal{S}}_k$ is lower bounded as 
\begin{align}
\sup\limits_{M\in\mathscr{F}_k(\alpha)} \mathbb{P}_M[\hat{\mathcal{S}}_k\neq \mathcal{S}_k^*]\geq \frac{1}{7},
\end{align}
where 
\begin{align*}
&g(n,m,\vec{\beta})=\frac{n}{m!}\sum_{t=1}^{q}(\beta_{t}-\beta_{m-q+t})A^{q-1}_{k-1}A_{n-k-1}^{m-q-1}\\
&h(n,m)= \frac{1}{{m \choose q}}\left[{k-1 \choose q-1}{n-k \choose m-q}+{k \choose q}{n-k-1 \choose m-q-1}\right],
\end{align*}
and $q\triangleq \max(1,m-n+k)$ and $A_{n}^k=n!/(n-k)!$.
\end{theorem}
Notice that $h(n,m)$ can also be written as $h(n,m)=\frac{nq+k(m-2q)}{m!}A^{q-1}_{k-1}A_{n-k-1}^{m-q-1}$. This means that the parameters $g(n,m,\vec{\beta})$ and $h(n,m)$ are roughly in the same order. 
In particular, when $m-n+k\leq1$, we have $q=1$. The parameters can be simplified further as
\begin{align*}
        &g(n,m,\vec{\beta})=\frac{n(\beta_1-\beta_m)A_{n-k-1}^{m-2}}{m!}\\
       & h(n,m)=\frac{(n+k(m-2))A_{n-k-1}^{m-2}}{m!}.
    \end{align*}

\begin{figure}[t]
\centering
\includegraphics[scale=0.6]{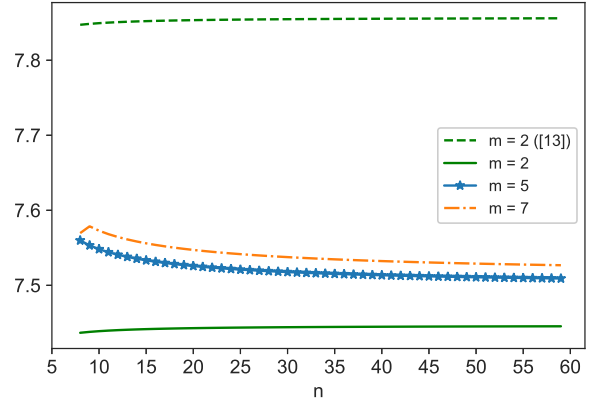}
\caption{The gap between the upper bound and the lower bound on $\alpha$ in Theorem \ref{theorem:forward2} and Theorem \ref{theorem:converse} for different $m$ when  $p=0.2$. The notch on the dashed line is where the value of $q$ changes from $m-n+k$ to $1$. \label{fig:bound_p=0.2}}
\end{figure}
\begin{figure}[h]
\centering
\includegraphics[scale=0.6]{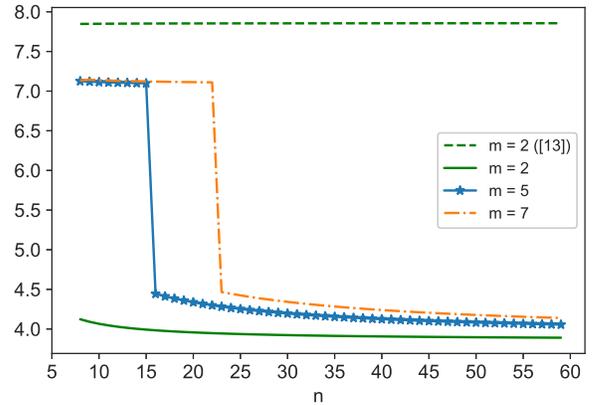}
\caption{The gap between the upper bound and the lower bound on $\alpha$ in Theorem \ref{theorem:forward2} and Theorem \ref{theorem:converse} for different $m$ when $p=0.4$. The drops on the dashed line and the line with asterisks markers are where the value of $p_0$ changes from $p_0\geq p $ to $p_0\leq p$.\label{fig:bound_p=0.4}}
\end{figure}

\begin{figure}[h]
\centering
\includegraphics[scale=0.6]{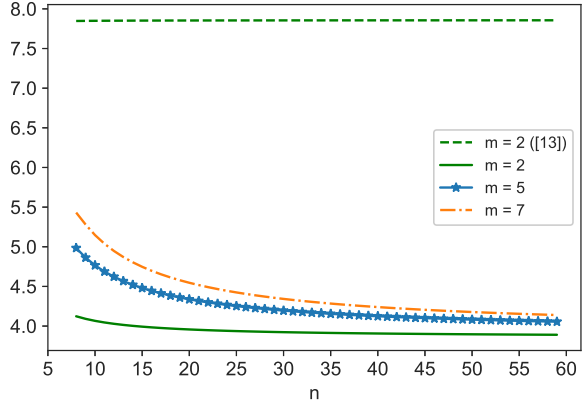}
\caption{ The gap between the upper bound and the lower bound on $\alpha$ in Theorem \ref{theorem:forward2} and Theorem \ref{theorem:converse} for different $m$ when $p=1$. \label{fig:bound_p=1}}
\end{figure}

In Figs. \ref{fig:bound_p=0.2},   \ref{fig:bound_p=0.4}, and  \ref{fig:bound_p=1}, we plot { the gap between the upper bound and the lower bound on $\alpha$ in Theorem \ref{theorem:forward2} and Theorem \ref{theorem:converse}} for different assignments of $p$. The score function $\vec{\beta}$ is chosen to be of uniform spacing and $\beta_m=0$, $k=3$. { The results in \cite{shah2017simple} are also included for reference.} It can be seen that the gap monotonically increases with $m$ in general. The gap between the forward direction and the converse direction is smaller for large observation probability $p$ in general, mainly because in this regime we have $p\geq p_0$, and the bound in Theorem \ref{theorem:forward2} is tighter. {From Figs. \ref{fig:bound_p=0.2},   \ref{fig:bound_p=0.4}, and  \ref{fig:bound_p=1}, it can be seen that the gap on $\alpha$ in our new results is greatly reduced compared with the results in \cite{shah2017simple} for the pairwise case.}

The proof of Theorem \ref{theorem:converse} is given in Section \ref{sec:proofconverse}, which relies on the application of Fano's inequality on a set of probability distributions that are difficult to distinguish for the top-$k$ items. 

\subsection{Analysis of the Error Probability Bounds}

To develop a better understanding of our main result, we return to the special case of $m=2$ and compare our new bound with the one derived in \cite{shah2017simple}. For this case, $(\beta_1,\beta_{2})=(1,0)$, and it is clear that $p\rho_{n,m}=p(n-1)$. We have the following observations. 
\begin{enumerate}
\item Theorem \ref{theorem:forward2} for $m=2$ is stronger than the bound in the forward direction in \cite{shah2017simple}. 
Theorem \ref{theorem:forward2} can be simplified to produce the following set 
\begin{align}
\mathscr{F}_k(\alpha)=\left\{M\in\mathcal{M}: \Delta_k\geq\alpha\sqrt{\frac{\log n}{rp(n-1)}}\right\}.
\end{align}
The result in \cite{shah2017simple} has a slightly different definition of $\Delta_k$ without the normalization, which can be directly translated into our notation, and it induces the following set $\mathscr{F}^{(1)}_k(\alpha)$ 
\begin{align*}
\mathscr{F}^{(1)}_k(\alpha)=\left\{M\in\mathcal{M}: \Delta_k\geq\alpha \sqrt{\frac{\log n}{rp(n-1)}}\sqrt{\frac{n}{n-1}}\right\}.
\end{align*}
The error probability bound given in \cite{shah2017simple} is that when $\alpha\geq 8$, 
\begin{align}
 \sup\limits_{M\in\mathscr{F}^{(1)}_k(\alpha)} \mathbb{P}_M[\tilde{\mathcal{S}}_k\neq \mathcal{S}_k^*]  \leq n^{-\alpha^2/4+2}\leq n^{-14}. \label{eqn:alpha8also}
\end{align}
It is easy to verify that $\mathscr{F}^{(1)}_k(\alpha)\subseteq\mathscr{F}_k(\alpha)$. Thus, in order to show Theorem \ref{theorem:forward2} is stronger than the result in \cite{shah2017simple} for $m=2$, we only need to show our error probability is smaller than the right hand side of (\ref{eqn:alpha8also}). For this purpose, we assume  $\alpha\geq8$. In the regime $0<p\leq p_0$, we have 
\begin{align}
k(n-k)n^{-\frac{\alpha^2}{(4-2p)}}\leq n^{-\alpha^2/4+2},
\end{align}
since $p>0$. In the regime $p_0<p\leq1$, and for the case where $n\geq3$, we have $p_0=1-\sqrt{1-\frac{n}{2(n-1)}}$. Then
\begin{align}
\label{ineq:outerBound,m=2}
&k(n-k)n^{-\frac{\alpha^2p}{(2(1-\beta_2)^2\frac{m-1}{n-1}+\frac{n-m}{n-1})}} \notag\\
& =k(n-k)n^{-\frac{\alpha^2}{2(1+\sqrt{1-\frac{n}{2(n-1)}})}}\notag\\
&\leq n^{-\alpha^2/4+2}.
\end{align}
When $n=2$, $p_0=\frac{1}{2}$, the bound in (\ref{ineq:outerBound,m=2}) also holds. We thus come to the conclusion that in both regimes, the bounds on the error probability in Theorem \ref{theorem:forward2} are stronger than that in \cite{shah2017simple}. 
\item The bound in Theorem \ref{theorem:converse} for $m=2$ is equivalent to the converse direction result given in \cite{shah2017simple}. 
When $m=2$, the assumptions in Theorem \ref{theorem:converse} reduce to $\alpha\leq\frac{1}{7}\sqrt{\frac{n}{n-1}}$,
and we  have $p\geq\frac{\log n}{2rn}$. It implies that the conditions and the bound given in Theorem 2 matches precisely those given in the corresponding converse in \cite{shah2017simple}, after taking into account the slightly different definition of $\Delta_k$ between ours and \cite{shah2017simple}.
\item Theorems \ref{theorem:forward2} and  \ref{theorem:converse} together translate to a reduced gap between the upper bound and the lower bound on $\alpha$ than the previous best result in \cite{shah2017simple}. When $p$ is not too small and $n>7$, the result given in \cite{shah2017simple} can be understood as follows: when $\alpha\leq 1/7$, the error probability of Borda counting-based method suffers from a large probability of error (larger than $1/7$), and when $\alpha\geq 8$, the error probability becomes small (less than $n^{-14}$, i.e., diminishing in $n$ and extremely small when $n$ is moderately large). Because of the refined bounds in Theorem \ref{theorem:forward2}, we can reduce the gap further when keeping the probability of error the same as in \cite{shah2017simple}. In the low observation probability regime, it is clear that when  $\alpha\geq 4\sqrt{4-2p}$ (i.e., a value strictly less than $8$), the error probability is less than $n^{-14}$; in the high probability regime where $p\geq p_0=1-\sqrt{1-\frac{n}{2(n-1)}}\geq 1-\sqrt{1/2}$, Theorem \ref{theorem:forward2} states that $\alpha\geq 4\sqrt{n/p(n-1)}$ suffices.  In both cases, the gaps between the upper bounds and the lower bounds on $\alpha$ are reduced, although they do not yet completely match. 
\item The asymptotic gap when $n\rightarrow \infty$ can be identified as follows. If we allow $\alpha$ to be chosen such that 
\begin{align} 
\sup\limits_{M\in\mathscr{F}_k(\alpha)} \mathbb{P}_M[\tilde{\mathcal{S}}_k\neq \mathcal{S}_k^*]\rightarrow 0 
\end{align} 
in the result of \cite{shah2017simple}, then it is clear that we need $\alpha>2\sqrt{2}$, while the lower bound remains $\alpha\leq 1/7$. With the result in Theorem \ref{theorem:forward2}, we can refine the upper bound to be 
\begin{align}
\alpha> \begin{cases}
2\sqrt{2-p}& 0<p\leq p_0\\
\sqrt{2/p}& p_0<p\leq 1
\end{cases},\label{eqn:asymp}
\end{align}
and thus in both regimes the asymptotic gap on $\alpha$ for the phase change is reduced. 
\end{enumerate} 

For general $m$, we can similarly seek the lower bounds on $\alpha$, such that the error probability is less than $n^{-\gamma}$.
\begin{enumerate}
    \item For fixed $\gamma$ (e.g., $\gamma=14$ as in the case $m=2$), it can be verified that 
    \begin{align}
        \alpha> \begin{cases}
\sqrt{(2+\gamma)(4-2p)}& 0<p\leq p_0\\
\sqrt{\frac{(2+\gamma)}{p}(2(1-\beta_m)^2\frac{m-1}{n-1}+\frac{n-m}{n-1})} & p_0<p\leq 1
\end{cases},
    \end{align}
    \item For the asymptotic case we can choose $\gamma\rightarrow 0$, then the same condition as in (\ref{eqn:asymp}) drives the error probability to zero when $n\rightarrow \infty$.  
\end{enumerate}

\label{specialcase:top-k}

\subsection{Proof Outline of Theorem \ref{theorem:forward2} }
\label{sec:proofoutline}

Theorem \ref{theorem:forward2} is essentially the combination of the two bounds in the two corresponding regimes, which are given as two propositions below.
\begin{prop}
\label{lemma:forward2}
For any $\alpha>0$ and $p\in(0,1]$, the probability of choosing incorrect top-$k$ items using the Borda counting-based method for any items with $M\in\mathscr{F}_k(\alpha)$ is upper bounded as
\begin{equation}
\label{eqn:forward2}
   \sup\limits_{M\in\mathscr{F}_k(\alpha)} \mathbb{P}_M[\tilde{\mathcal{S}}_k\neq \mathcal{S}_k^*]\leq k(n-k)n^{-\frac{\alpha^2p}{(2(1-\beta_m)^2\frac{m-1}{n-1}+\frac{n-m}{n-1})}}.
\end{equation}
\end{prop}
\begin{prop}
\label{lemma:forward1}
For any $\alpha>0$ and $p\leq\frac{1}{2}$, the probability of choosing incorrect top-$k$ items using the Borda counting-based method for any items with $M\in\mathscr{F}_k(\alpha)$ is upper bounded as
\begin{equation}
\label{eqn:forward1}
   \sup\limits_{M\in\mathscr{F}_k(\alpha)} \mathbb{P}_M[\tilde{\mathcal{S}}_k\neq \mathcal{S}_k^*]\leq k(n-k)n^{-\frac{\alpha^2}{(4-2p)}}.
\end{equation}
\end{prop}

Detailed proofs of these propositions can be found in Section \ref{sec:prooflemmas}. We are now ready to prove Theorem \ref{theorem:forward2}. 
\begin{proof}[Proof of Theorem \ref{theorem:forward2}] 
Notice that the statement that the bound in Proposition \ref{lemma:forward1} is tighter than that in
Proposition \ref{lemma:forward2} is equivalent to
\begin{align}
\frac{p}{(2(1-\beta_m)^2\frac{m-1}{n-1}+\frac{n-m}{n-1})}\geq \frac{1}{(4-2p)},
\end{align}
which can be simplified to 
\begin{align}
   p^2-2p+\frac{(2(1-\beta_m)^2\frac{m-1}{n-1}+\frac{n-m}{n-1})}{2}>0.
\label{eqn: compare two lemmas}
\end{align}
Since $p\in[0,1]$, we have $p\leq1-\sqrt{1-(1-\beta_m)^2\frac{m-1}{n-1}-\frac{1}{2}\frac{n-m}{n-1}}$. Combining the condition $p\leq \frac{1}{2}$ in Lemma \ref{lemma:forward1}, we obtain the result that when 

\begin{align}
p\leq p_0=\min\left\{\frac{1}{2},1-\sqrt{1-(1-\beta_m)^2\frac{m-1}{n-1}-\frac{1}{2}\frac{n-m}{n-1}}\right\}, 
\end{align}
the bound in Lemma \ref{lemma:forward1} holds and it is tighter than that in Lemma \ref{lemma:forward2}. On the other hand, when $p\geq p_0$, the bound in Lemma \ref{lemma:forward2} holds and is tighter. This completes the proof.
\end{proof}

\section{Approximate Top-$k$ Selection}
\label{section: hamming error}
In the previous section, we consider the exact top-$k$ selection problem. In many cases, it may be sufficient to identify the top-$k$ set approximately, which we study in this section. This approximate setting has also been considered previously in \cite{shah2017simple} for the special case of $m=2$.

\subsection{Main Results}

First, we define
\begin{align}
\mathscr{F}_{k,h}(\alpha)=\left\{M\in\mathcal{M}: \Delta_{k,h}\geq\alpha\sqrt{\frac{\log n}{rp\rho_{n,m}}}\right\},\label{eqn:Falphakh}
\end{align}
where $\Delta_{k,h} :=\tau_{(k-h)}-\tau_{(k+h)}.$ Comparing $\mathscr{F}_{k,h}(\alpha)$ with $\mathscr{F}_{k}(\alpha)$ in (\ref{eqn:Falpha}), it can be seen that the only difference is that $\Delta_{k}$ in (\ref{eqn:Falpha}) is replaced by $\Delta_{k,h}$, i.e., it was for the case $\Delta_{k,0}=\Delta_{k}$. We denote the Hamming distance between two subsets $A, B$ of $[n]$ as
 \begin{align}
 D_H(A, B)= \left \vert \left\{A\cup B\right\}\setminus\left\{ A\cap B\right\} \right \vert. 
 \end{align}
\begin{theorem}
\label{theorem:hamming_forward}
For any $\alpha>0$, we have
\begin{align}
&\sup\limits_{M\in\mathscr{F}_{k,h}(\alpha)} \mathbb{P}_M[D_H(\tilde{\mathcal{S}}_k,
       \mathcal{S}_k^*)>2h]\notag\\
&\leq
\begin{cases}
  (k-h)(n-k-h)n^{-\frac{\alpha^2}{(4-2p)}}& 0<p\leq p_0\\
  (k-h)(n-k-h)n^{-\frac{\alpha^2p}{(2(1-\beta_m)^2\frac{m-1}{n-1}+\frac{n-m}{n-1})}}& p_0<p\leq1
\end{cases},
\end{align}
where $p_0$ has the same definition as in Theorem \ref{theorem:forward2}, $\tilde{\mathcal{S}}_k$ is the top-$k$ estimate generated by the Borda-counting algorithm.
\end{theorem}

\begin{theorem}
\label{theorem:hamming_converse}
Let $\nu_1$, $\nu_2 \in (0, 1)$ be two constants such that $2h\leq\frac{1}{1+\nu_2}\min\left\{n-k, k, n^{1-\nu_1}\right\}$. In the regime $p\geq\frac{\log
n}{4h(n,m)r}$, for any $\alpha\leq \frac{\sqrt{2}}{14}g(n,m,\vec{\beta})\sqrt{\frac{\nu_1\nu_2}{h(n,m)\rho_{n,m}}}$, any estimator $\hat{\mathcal{S}}_k$ has an error probability lower bounded by
   \begin{equation}
       \sup\limits_{M\in\mathscr{F}_{k,h}(\alpha)} \mathbb{P}_M[D_H(\hat{\mathcal{S}}_k,
       \mathcal{S}_k^*)>2h]\geq\frac{1}{7},
   \end{equation}
for any $n$ that is larger than a $(\nu_1, \nu_2)$-dependent constant  $c(\nu_1,\nu_2)$, and $q\triangleq \max(1,m-n+k).$ 
\end{theorem}

Consider again the special case $m=2$: it can be verified that Theorem \ref{theorem:hamming_forward} is slightly stronger than the corresponding result in \cite{shah2017simple}, and Theorem \ref{theorem:hamming_converse} is precisely the same. 

\subsection{Proof Outlines of Theorems \ref{theorem:hamming_forward} and \ref{theorem:hamming_converse}}

We denote the upper bound of ranking item $a$ lower than item $b$ as $P$.
Theorem \ref{theorem:hamming_forward} is in fact a direct consequence of Theorem \ref{theorem:forward2}.  To see this, note that 
 Theorem \ref{theorem:forward2} implies that with probability at least $1-(k-h)(n-k-h)P$, $\tilde{\mathcal{S}_k}$ ranks every item in $\left\{1,...,k-h\right\}$ higher than every
item in the set $\left\{k+h+1,...,n\right\}$. 
Thus, we have either $\tilde{\mathcal{S}_k}\subseteq [k + h]$ or $[k-h] \subseteq \tilde{\mathcal{S}_k}$.  Either 
case would lead to $|\tilde{\mathcal{S}_k}\cap[k]| \geq k-h$, thereby proving Theorem \ref{theorem:hamming_forward}.

To prove Theorem \ref{theorem:hamming_converse}, we follow the approach in \cite{shah2017simple}. The following lemma is instrumental in the proof.

\begin{lemma}[Shah et al. \cite{shah2017simple}]
 \label{lemma:ksized} 
In the regime $2h\leq \frac{1}{\nu_1+\nu_2} \min\left\{n^{1-\nu_1} , k, n-k\right\}$ for some constants $\nu_1 \in (0, 1)$ and $\nu_2 \in (0, 1)$, and when n is larger than a $(\nu_1, \nu_2)$-dependent constant, there exists a subset ${b_1, . . . , b_L}\subseteq{\left\{0, 1\right\}}^{n/2}$
with cardinality $L \geq L^*\triangleq e^{\frac{9}{10} \nu_1\nu_2h\log n}$, such that
    $$D_H(b_j,\textbf{0})=2(1+\nu_2)h, \quad D_H(b_j,b_k)>4h$$ 
for all $j\neq k\in[L]$.
\end{lemma}

It was shown in \cite{shah2017simple} that constructing $L^*$ different probability distributions satisfying the two properties below would prove Theorem \ref{theorem:hamming_converse} when $m=2$. 
\begin{enumerate}
\setlength{\itemsep}{1pt}
    \item For every $i \in [L^*]$, let $\mathcal{S}_k^{i}\subseteq[n]$ denote the set of top-$k$ items under the $i$-th distribution. Then for every $k$-sized set $\mathcal{S} \in [n]$,
\begin{equation}
    \sum_{i=1}^L\mathbb{1}({D_H(\mathcal{S}, \mathcal{S}_k^{i}) \leq 2h}) \leq 1.\label{eqn:uniquedecode}
\end{equation}

\item If the underlying distribution is chosen uniformly at random from this set of $L^*$ distributions, then any estimator that attempts to identify the underlying distribution
$i \in [L^*]$ errs with probability at least $1-\frac{1}{7}$.\\
\end{enumerate}

We follow the approach in \cite{shah2017simple} to construct $L^*$ different $k$-sized subset of $[n]$.  For each $i \in [L^*]$, let $B_{i}$ denote a $[2(1 + \nu_2)h]$-sized subset of
$\left\{\frac{n}{2} + 1,...,\frac{n}{2}\right\}$. The items in $B_{i}$ correspond to the $2(1 + \nu_2)h$ positions being 1 in the $i$-th string as specified by Lemma \ref{lemma:ksized}. Define the sets $A_{i}=\{1,...,k-2(1+\nu_2)h \}$. The $k$-sized subset $\mathcal{S}^{i}_k$ is then constructed as $\mathcal{S}^{i}_k=A_{i}\cup B_{i}$, which is valid since $2h\leq \frac{1}{1+\nu_2}k $. By Lemma \ref{lemma:ksized}, for
any distinct $i_1, i_2 \in [L]$, we have $D_H(A_{i_1} \cup B_{i_1}, A_{i_2} \cup B_{i_2}) \geq
4h+1$. This implies that the condition (\ref{eqn:uniquedecode}) is satisfied, since otherwise the existence of a set $\mathcal{S}$ that makes the LHS of (\ref{eqn:uniquedecode}) greater than $1$ would make the two corresponding sets $\mathcal{S}^{i_1}_k$ and $\mathcal{S}^{i_2}_k$ differ by strictly less than $(4h+1)$.

We now construct $L^*$ probability distributions. For every $i \in [L^*]$ and the $k$-sized subset $\mathcal{S}^{i}_k$ given above, let the distribution $M^{i}_{\vec{v}}$ be as in (\ref{eqn:worsedist}), 
with the conditions replaced by 
\begin{itemize}
    \item C1: $\vec{v}_{1:q}\subseteq S_k^{i} \text{ and }\vec{v}_{q+1:m}\subseteq [n]\setminus S_k^{i}$;
    
    \item C2: $\vec{v}_{1:m-q}\subseteq
    [n]\setminus S_k^{i}\text{and } \vec{v}_{m-q+1:m}\subseteq S_k^{i}$.
\end{itemize}
This leads to 
\begin{equation}
    \Delta_{k,h}=\Delta_k=\frac{g(n,m,\vec{\beta})}{\rho_{n,m}}\delta \label{threshold and delta}.
\end{equation}

For any $i_1,i_2\in[L^* ]$, an upper bound for the Kullback-Leibler divergence between two probability distributions is needed. For any $\mathcal{A}\subseteq[n]$ where $|\mathcal{A}|=m$, $\mathbb{P}^{i_1}(V_{\mathcal{A}}^{(\ell)})\neq\mathbb{P}^{i_2}(V_{\mathcal{A}}^{(\ell)})$ only if $\exists j\in\mathcal{A}$, such that $j\in B_{i_1}\cup B_{i_2}$. In Theorem \ref{theorem:converse}, we have shown that $D_{KL}(\mathbb{P}^{i_1}(V_{\mathcal{A}}^{(\ell)})||\mathbb{P}^{i_2}(V_{\mathcal{A}}^{(\ell)}))\leq {\binom{m}{q}}^{(-1)}\frac{8\delta^2}{1-\delta^2}$, from which it follows that
\begin{align}
    D_{KL}(\mathbb{P}^{i_1}||\mathbb{P}^{i_2})\leq 4prh(n,m)(1+\nu_2)h\frac{8\delta^2}{1-\delta^2}.
\end{align}
Under the assumptions on $\alpha$ and $p$, it can be verified straightforwardly that $\delta\leq\frac{\sqrt{2}}{7}$, from which it follows that 
\begin{equation}
    D_{KL}(\mathbb{P}^{i_1}||\mathbb{P}^{i_2})\leq \frac{3}{4}\nu_1\nu_2h\log n.
\end{equation}
Suppose that the underlying distribution is drawn uniformly at random from the constructed set, and the index is $i^*$. Using Fano's inequality, any estimator $\hat{i^*}$ must have error probability lower bounded by
\begin{align}
    \mathbb{P}(\hat{i^*}\neq i^*)\geq
    1-\frac{\frac{3}{4}\nu_1\nu_2 h \log n+\log2} { \frac{9}{10} \nu_1\nu_2h\log n}\geq\frac{1}{7}.
\end{align}
This establishes property 2) for the distributions. The proof of Theorem \ref{theorem:hamming_converse} can now be completed straightforwardly along the same line as in \cite{shah2017simple}.

\section{Numerical Results}
\label{section: experiment}
In this section, we evaluate and compare the performances of three algorithms, including the proposed Borda counting-based algorithm, via numerical simulations on both synthetic data and real-world data.
\subsection{Algorithms for  Comparison}
We compare three algorithms:
\begin{enumerate}
    \item \textbf{Borda counting-based algorithm}, which is the algorithm studied in this paper. 
    In the experiments we need to specify the score functions $\beta$, and we consider $m+4$ different score functions as given in (\ref{eqn:beta assignment}). The subscripts in the definitions represent the positions of the items in the ranking. For example, $\bar{\beta}^{(i)}_j$ is the score assigned to the item ranked in the $j$-th position under the score function $\bar{\beta}^{(i)}$. Among these score functions, $\tilde{\beta}^{(i)}$ has the first $i$ positions set as $1$, and other positions set to zero. The score functions $\check{\beta}^{(1)}$ and $\check{\beta}^{(2)}$ are inversely proportional functions of the ranking positions $j$, the score functions $\bar{\beta}^{(1)}$ and $\bar{\beta}^{(2)}$ are in linear forms, and $\hat{\beta}$ is in a quadratic form.    
    \begin{align}
    \label{eqn:beta assignment}
         &\tilde{\beta}^{(i)}_j=\mathbbm{1}(j\leq i),\quad  j\in[m],\quad i\in[m-1]\notag\\
         &\check{\beta}^{(1)}_j=\frac{2m-2}{j+m-2}-1,\quad j\in[m]\notag\\
         &\check{\beta}^{(2)}_j= \frac{1}{j},\quad j\in[m]\notag\\
         &\bar{\beta}^{(1)}_j= 1-\frac{j-1}{m-1},\quad  j\in[m]\notag\\
           &\bar{\beta}^{(2)}_j= 1-\frac{j-1}{m},\quad  j\in[m]\notag\\
         &\hat{\beta}_j= \frac{m^2-j^2}{m^2-1},\quad  j\in[m].
    \end{align}
     \item \textbf{Borda counting-based algorithm with normalization}, which is similar to the first method, with only one difference. Here the score obtained from the comparisons of the same subset is normalized by the number of observed comparisons for this subset. More precisely, the score of any item $a$ is 
    \begin{equation}
      S_a=\sum_{\mathcal{A}^-\subseteq{[n]}\setminus \{a\}}\big(\frac{1}{Z_{\mathcal{A}}+1}\sum_{\ell\in{[r]}}X_{a,\mathcal{A}^-}^{(\ell)}\big),
\end{equation}
where $Z_{\mathcal{A}}$ is the number of observed comparisons for the set $\mathcal{A}$.
    \item \textbf{Spectral MLE-based algorithm}, which was proposed in \cite{jang2017optimal}. The spectral MLE based algorithm converts the $m$-wise comparison data into pairwise comparison data, and then uses the pairwise spectral MLE algorithm on the generated data to select the top-$k$ items.
\end{enumerate}

\subsection{Results with Synthetic Data}
We generate the data with $n=10$ total items and $m=4$ items participating in each comparison. The underlying probability distribution is based on the well-known Plackett-Luce (PL)  model and a noisy variant of the model, as described in the following.
\begin{figure}
\centering
  \includegraphics[width=.85\linewidth]{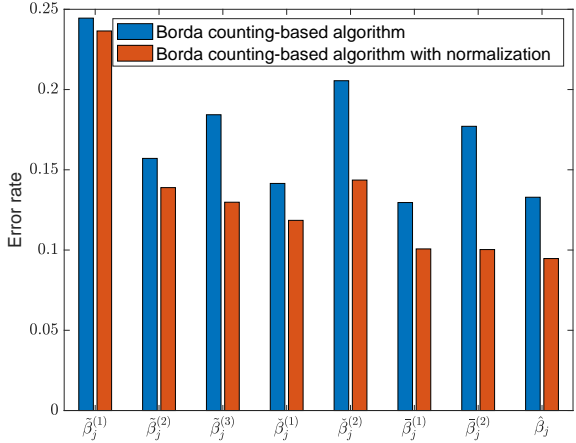}
  \caption{Error rates of different $\beta$ assignments}
  \label{fig:synthetic1_sub2}
\end{figure}
\begin{enumerate}
    \item \textbf{Plackett-Luce model.} See Definition \ref{def:PL}
    We define the weight vector $\vec{w}=(w_1,w_2,...,w_n)$, where $w_i$ is the weight of item $i$. 
    We use the two kinds of weight assignments  $\vec{w}^{(1)}$ and $\vec{w}^{(2)}$ in our simulation:
    \begin{align}
        \label{eqn:s}
        w_i^{(1)}&=15+i,\quad i\in[n]\notag\\
        w_i^{(2)}&=1.1^i,\quad i\in[n].
    \end{align}
    {Notice that by Lemma~\ref{lemma:consistency of para model}, the ranking given by the associated scores is indeed consistent with the ranking of the underlying weight vector, therefore, it can be viewed as the ground truth. }
    \item \textbf{PL model with noise.} We create the underlying distribution by first applying the PL model with weights in equation (\ref{eqn:s}), and then we add some Gaussian noise for each permutation of each $m$-wise group.
    \begin{equation}
        M_{v_1,v_2,...,v_m}=M_{\vec{v}}=\max\left\{\prod_{k=1}^m f_k(\vec{v})+\epsilon_{\vec{v}},0\right\},
    \end{equation}
    where $\epsilon_{\vec{v}} \; \sim_{ i.i.d}  \;  \mathcal{N}(0,\sigma^2)$. Normalization is then enforced to ensure $\sum_{\vec{v}\doteq\mathcal{A}}M_{v_1,v_2,...,v_m}=1$.
\end{enumerate}

\begin{figure}
\label{fig:synthetic2}
\centering
\begin{subfigure}{.48\textwidth}
  \centering
  \includegraphics[width=.85\linewidth]{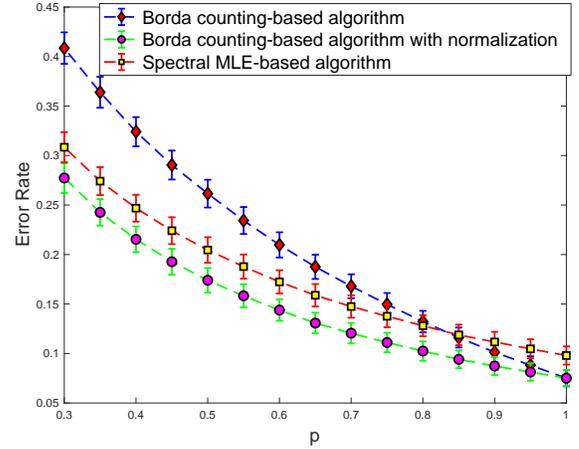}
  \caption{Results of different $p$ (PL model)}
  \label{fig:synthetic2_sub1}
\end{subfigure}
\begin{subfigure}{.48\textwidth}
  \centering
  \includegraphics[width=.85\linewidth]{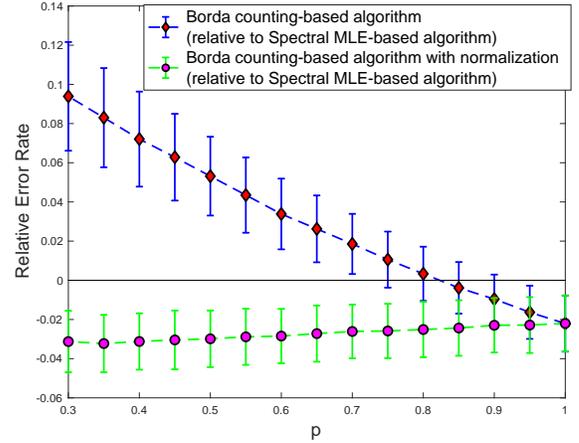}
  \caption{Results of different $p$ (PL model with noise, $\sigma=0.025$)}
  \label{fig:synthetic2_sub2}
\end{subfigure}
\caption{Error rates of three algorithms under different value of $p$ with underlying weight $\vec{w}=\vec{w}^{(1)}$.}

\end{figure}

\begin{figure}
\label{fig:synthetic2}
\centering
\begin{subfigure}{.48\textwidth}
  \centering
  \includegraphics[width=.85\linewidth]{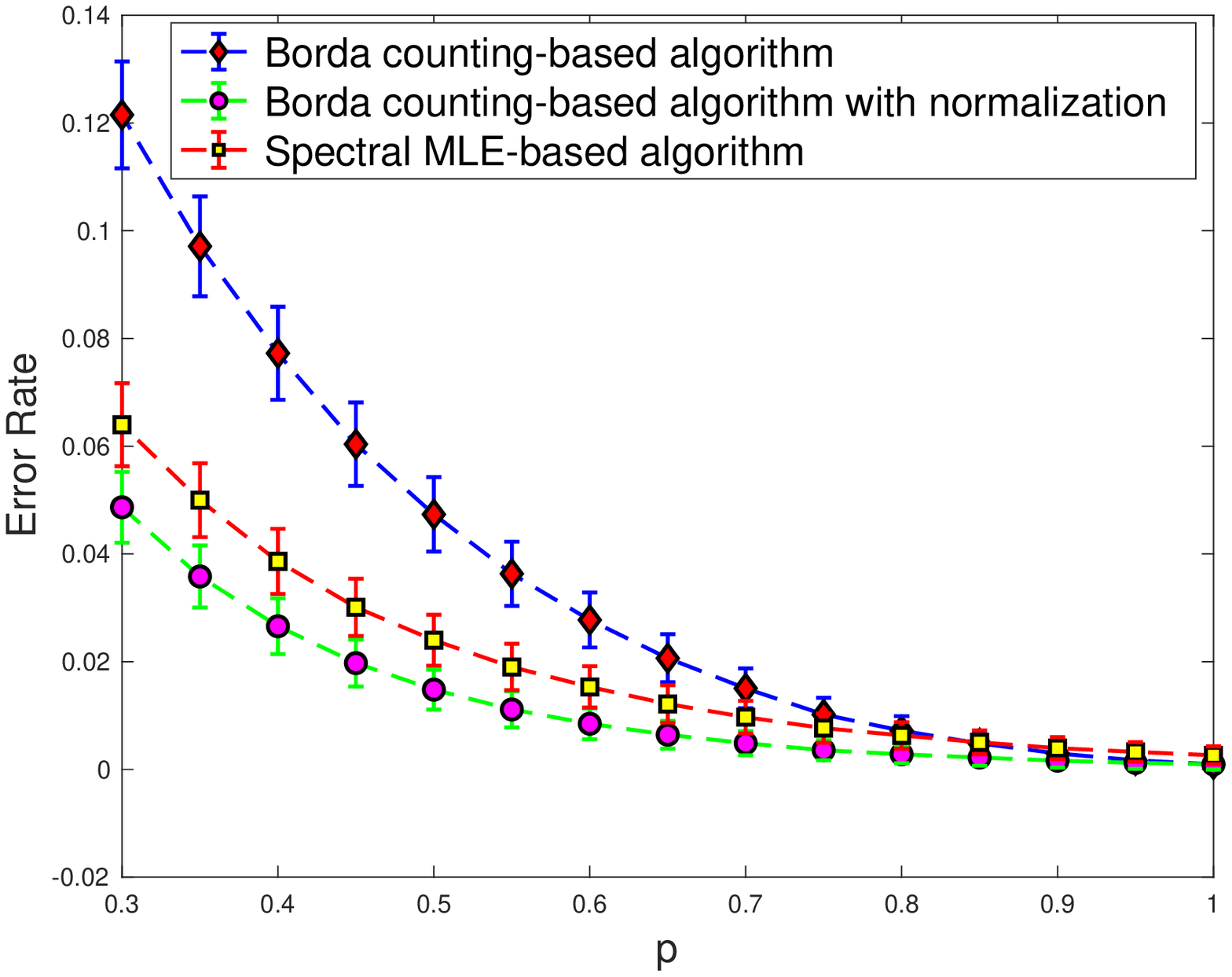}
  \caption{Results of different $p$ (PL model)}
  \label{fig:p2_PL}
\end{subfigure}
\begin{subfigure}{.48\textwidth}
  \centering
  \includegraphics[width=.85\linewidth]{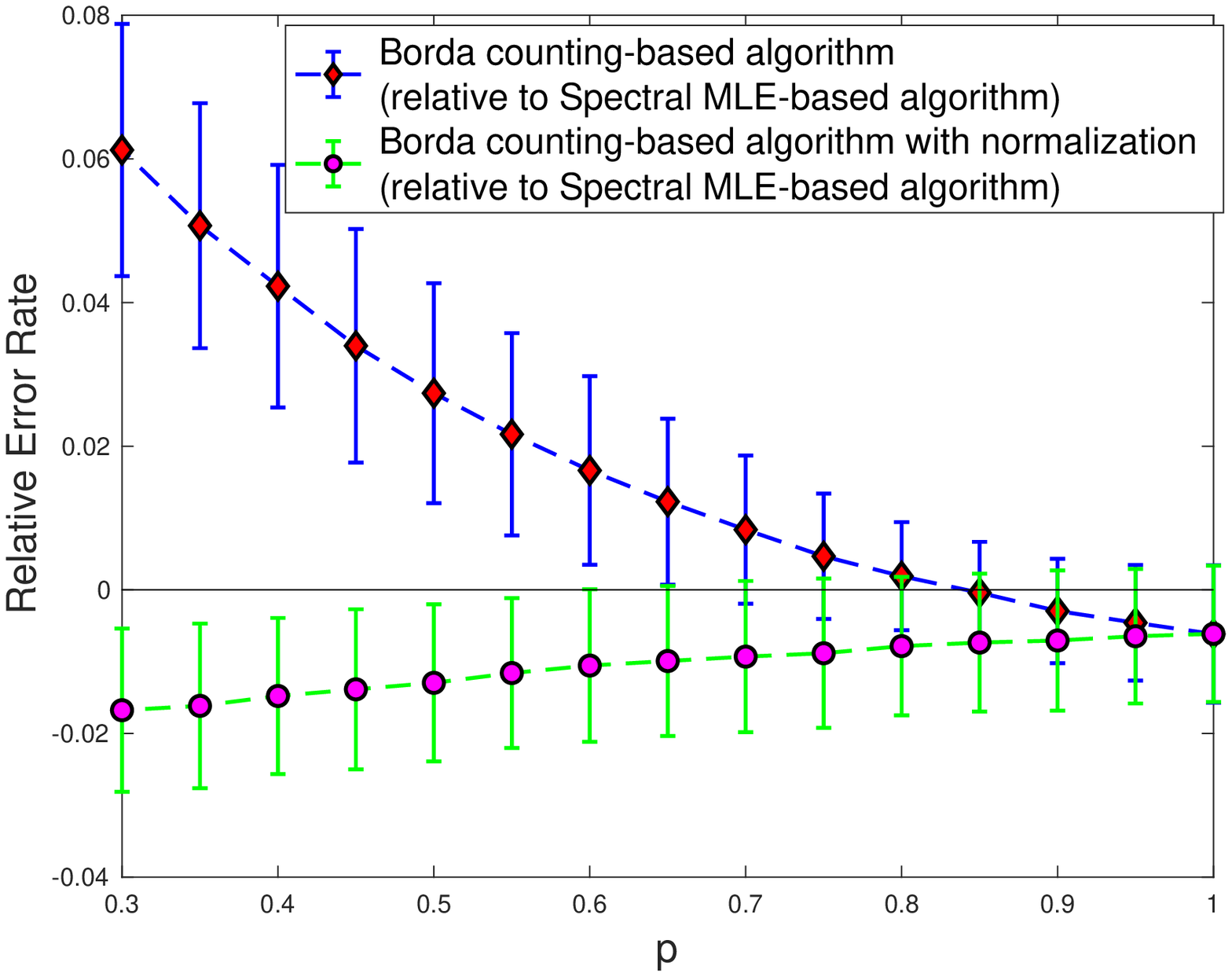}
  \caption{Results of different $p$ (PL model with noise, $\sigma=0.025$)}
  \label{fig:p2_PL with noise}
\end{subfigure}
\caption{Error rates of three algorithms under different value of $p$ with underlying weight $\vec{w}=\vec{w}^{(2)}$. }

\end{figure}

We first evaluate the performance of Borda counting-based algorithm and Borda counting-based algorithm with normalization under different assignments of score functions $\beta$ as in equations (\ref{eqn:beta assignment}). We generate the partial ranking data with observation probability $p=0.8$, and the total number of comparisons for each $m$-wise group $r=50$. The underlying distribution follows the PL model with weight vector $\vec{w}=\vec{w}^{(1)}$. 
The simulation results are presented in Figure \ref{fig:synthetic1_sub2}. It can be seen that the score assignment $\bar{\beta}^{(1)}$ has a lower error rate than others, and thus we  fix $\bar{\beta}^{(1)}$ in the rest of simulations. 

We then compare the performance of the three algorithms as functions of $p$ using the PL model with the underlying weight assignments  $\vec{w}=\vec{w}^{(1)}$ and $\vec{w}=\vec{w}^{(2)}$, respectively. The results are given in Figures \ref{fig:synthetic2_sub1} and \ref{fig:p2_PL}. Correspondingly, the results of using the PL model with Gaussian noise are presented in Figures \ref{fig:synthetic2_sub2} and \ref{fig:p2_PL with noise}. 
For each $p$  value, every $100$ sample trials were independently run and the errors are averaged over these runs. { We run a total of $1000$ such experiments for each $p$. The averaged error rate and standard error are calculated over these $1000$ runs.  For PL model with noise, we plot the relative error rate with respect to the spectral MLE-based algorithm}.

It is known that the spectral MLE algorithm is suitable and order-wise optimal for the BTL model \cite{chen2015spectral}, and thus we expect the spectral MLE-based algorithm to perform well for the PL model, which is a generalized version of the BTL model. From the numerical result, we first observe that the performance of Borda counting-based algorithm becomes more competitive to the spectral MLE-based algorithm when $p$ is large, but less so when $p$ is small. This behavior appears to being caused by the larger variation in the number of observed partial rankings for different cases when $p$ is small. This variation worsens the performance of the Borda counting-based algorithm relatively more than the spectral MLE-based algorithm\footnote{The Borda counting-based algorithm with normalization is in fact motivated by this observation.}. It can be seen that by incorporating normalization in the Borda counting-based algorithm, it is able to remedy this deficiency at smaller $p$ values. With the noisy PL model, it can be seen that the performances of all three algorithms are worse than under the noiseless PL model. However, the relative performances among the three methods do not change significantly. It should be noted that the normalized version of the Borda counting-based algorithm is considerably more complex than the Borda counting-based algorithm itself, since instead of maintaining only the accumulated scores for the items, we must also maintain the number of comparisons for each subset. This increase in complexity is similarly needed for the spectral MLE-based algorithm. 

\subsection{Results with Real-world Data}
The real-world data are collected from the \texttt{PrefLib} website \cite{sushi}. 
This dataset (sushi data) contains 5000 complete rankings of 10 kinds of sushi. Here we use the process below to find the ground truth: first run the Borda counting-based algorithm on the entire dataset with $m=10$ and the score function $\beta=\bar{\beta}^{(1)}$ as in Eqn.~\eqref{eqn:beta assignment},  and then sort the total score of each item to have the complete ranking.

For each batch size $i$ and each trial, we randomly choose $i$ complete rankings from the dataset as our mini-batch. In the comparison, since the comparison data are complete rankings of $n$ items in the sushi dataset, we extract the $m$-wise comparison data from the dataset as follows. For each complete ranking in the dataset, we randomly select an $m$-wise group among the total $n$ items, and then use the ranking result of the chosen items as our $m$-wise comparison data. Then we run the three algorithms on the generated $m$-wise data of the chosen mini-batch to determine if the top-$k$ estimate is correct. In the experiment, we choose $m=7$, $k=3$, and the score function $\beta=\bar{\beta}^{(1)}$ with $m=7$. {We run $100$ trials each time to calculate the error rate. The eventual error rate and standard error are calculated over all the error rates in $1000$ runs. }The simulation results are plotted in Figure \ref{fig:real}.  Since the real-world data is unlikely to follow the PL model, we expect the Borda counting-based methods to perform better than the spectral MLE-based method. The numerical results indeed confirm this expectation. 

\begin{figure}
\centering
\includegraphics[width=.85\linewidth]{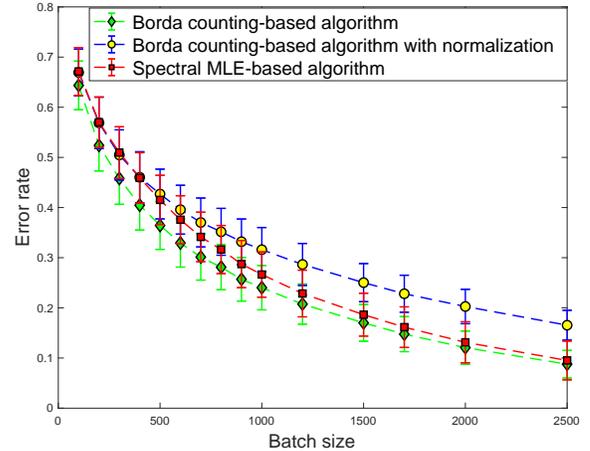}
\caption{Error rates of three algorithms under different batch size  $i$. }
\label{fig:real}
\end{figure}

\section{Proof Details of Theorem \ref{theorem:forward2}}
\label{sec:prooflemmas}

As previously discussed in the proof outline, the proof of Theorem \ref{theorem:forward2} boils down to proving Proposition \ref{lemma:forward2} and Proposition \ref{lemma:forward1}. In this section, we provide the proofs of these two propositions. 

\begin{proof}[Proof of Proposition \ref{lemma:forward2}]
Consider any item $a\in\mathcal{S}_k^*$ and $b\in[n]\setminus\mathcal{S}_k^*$, then define the event that $W_b> W_a$ as $\mathcal{E}_{ba}$, i.e.,
\begin{align}
\mathbf{Pr}(\mathcal{E}_{ba})=\mathbf{Pr}(W_b-W_a> 0). 
\end{align}
We start by noting that
\begin{align*}
W_b-W_a
=&\sum_{\ell\in{[r]}}\bigg{(}\sum_{\mathcal{A}^-\subseteq[n]\setminus\{b\}}X_{b,\mathcal{A}^-}^{(\ell)}-\sum_{\mathcal{A}^-\subseteq[n]\setminus\{a\}}X_{a,\mathcal{A}^-}^{(\ell)}\bigg{)}.
\end{align*}
It follows that
\begin{align}
&W_b-W_a=\sum_{\ell\in{[r]}}\sum_{\mathcal{A}^-\subseteq[n]\setminus\{a,b\}}\left({X}_{b,\mathcal{A}^-}^{(\ell)}-{X}_{a,\mathcal{A}^-}^{(\ell)}\right)\nonumber\\
&\quad-\sum_{\ell\in{[r]}}\sum_{\mathcal{A}^{--}\subseteq [n]\setminus\{a,b\}}\left(X_{a,\{b,\mathcal{A}^{--}\}}^{(\ell)}-X_{b,\{a,\mathcal{A}^{--}\}}^{(\ell)}\right).
\end{align}
In order to apply the concentration inequalities, define the centralized scores. For any $a\in [n]$
\begin{align*}
\bar{X}_{a,\mathcal{A}^-}^{(\ell)}&\triangleq X_{a,\mathcal{A}^-}^{(\ell)}-\mathbb{E}\left(X_{a,\mathcal{A}^-}^{(\ell)}\right)=X_{a,\mathcal{A}^-}^{(\ell)}-p\sum_{t=1}^{m}\beta_{t}R_{a,\mathcal{A}^-}(t),
\end{align*}
and the centralized cross-score
\begin{align}
&\bar{X}_{\{a,b\},\mathcal{A}^{--}}^{(\ell)}\triangleq X_{a,\{b,\mathcal{A}^{--}\}}^{(\ell)}-X_{b,\{a,\mathcal{A}^{--}\}}^{(\ell)}\nonumber\\
&\qquad\qquad-\mathbb{E}\left(X_{a,\{b,\mathcal{A}^{--}\}}^{(\ell)}\right)+\mathbb{E}\left(X_{b,\{a,\mathcal{A}^{--}\}}^{(\ell)}\right)\nonumber\\
&=X_{a,\{b,\mathcal{A}^{--}\}}^{(\ell)}-X_{b,\{a,\mathcal{A}^{--}\}}^{(\ell)}\nonumber\\
&\quad -\left(p\sum_{t=1}^{m}\beta_{t}R_{a,\{b,\mathcal{A}^{--}\}}(t)-p\sum_{t=1}^{m}\beta_{t}R_{b,\{a,\mathcal{A}^{--}\}}(t)\right).
\end{align}
It follows after some simple algebra that
\begin{align}
 &W_b-W_a=\sum_{\ell\in{[r]}}\sum_{\substack{\mathcal{A}^-\subseteq[n]\setminus\{a,b\}\\|\mathcal{A}^-|=m-1}}\left(\bar{X}_{b,\mathcal{A}^-}^{(\ell)}-\bar{X}_{a,\mathcal{A}^-}^{(\ell)}\right)\notag\\
 &\qquad-\sum_{\ell\in{[r]}}\sum_{\substack{\mathcal{A}^{--}\subseteq [n]\setminus\{a,b\}\\|\mathcal{A}^{--}|=m-2}}\bar{X}_{\{a,b\},\mathcal{A}^{--}}^{(\ell)}\nonumber\\
&\qquad +r\sum_{\substack{\mathcal{A}^-\subseteq[n]\setminus\{b\}\\|\mathcal{A}^-|=m-1}}p\sum_{t=1}^{m}\beta_{t}R_{b,\mathcal{A}^-}(t)\notag\\
&\qquad -r\sum_{\substack{\mathcal{A}^-\subseteq[n]\setminus\{a\}\\|\mathcal{A}^-|=m-1}}p\sum_{t=1}^{m}\beta_{t}R_{a,\mathcal{A}^-}(t)\nonumber,\\
\end{align}
\begin{align}
W_b-W_a=\Delta \bar{X}_{a,b}-rp\rho_{n,m}(\tau_a-\tau_b),
\end{align}
where for notation simplicity we have defined 
\begin{align}
\label{eqn: sum of var}
\Delta \bar{X}_{a,b}
&\triangleq \sum_{\ell\in{[r]}}\sum_{\substack{\mathcal{A}^-\subseteq[n]\setminus\{a,b\}}}\left(\bar{X}_{b,\mathcal{A}^-}^{(\ell)}-\bar{X}_{a,\mathcal{A}^-}^{(\ell)}\right)\nonumber\\
&\qquad\qquad-\sum_{\ell\in{[r]}}\sum_{\substack{\mathcal{A}^{--}\subseteq [n]\setminus\{a,b\}}}\bar{X}_{\{a,b\},\mathcal{A}^{--}}^{(\ell)}.
\end{align}
Because $a$ is in the top-$k$ and $b$ is not in the top-$k$, we have
\begin{align}
rp\rho_{n,m}(\tau_a-\tau_b)\geq rp\rho_{n,m}\Delta_k.
\end{align}
Since $W_b-W_a\geq 0$ is equivalent to  $\Delta \bar{X}_{a,b}\geq rp\rho_{n,m}(\tau_a-\tau_b)$, we have
\begin{align}
\label{ineq:probability of separated by delta}
\mathbf{Pr}(W_b-W_a\geq 0)\leq \mathbf{Pr}\left[\Delta \bar{X}_{a,b}\geq rp\rho_{n,m}\Delta k\right].
\end{align}
In order to bound the probability on the right hand side, we need the Hoeffding's inequality, which is stated in the lemma below for completeness.
\begin{lemma}[Hoeffding's inequality]\label{lemma:Hoeffding}
Let $X_1, ..., X_n$ be independent random variables with the empirical mean $\bar{X} = \frac{1}{n}(X_1 + \cdots + X_n).$ When it is known that $X_i$ are strictly bounded by the intervals $[a_i, b_i]$, we have
\begin{align}
\label{ineq: Hoeffding}
\mathbf {Pr} \left({\bar{X}}-\mathbb {E} \left[{\bar{X}}\right]\geq t\right)\leq \exp \left(-{\frac {2n^{2}t^{2}}{\sum _{i=1}^{n}(b_{i}-a_{i})^{2}}}\right)
\end{align}
 \end{lemma}
 
Now we view each centralized score and centralized cross-score in (\ref{eqn: sum of var}) as the random variables $X_i$ in the Hoeffding's inequality. It is clear that these centralized scores and centralized cross-scores are independent, which allows us to apply the Hoeffding's inequality. Moreover, they are zero-mean. These random variables have the following lower bounds and upper bounds, due to their definitions.
\begin{align}
-\mathbb{E}X_{a,\mathcal{A}^-}&\leq\bar{X}_{a,\mathcal{A}^-}\leq 1-\mathbb{E}X_{a,\mathcal{A}^-}\nonumber\\
-\mathbb{E}X_{b,\mathcal{A}^-}&\leq\bar{X}_{b,\mathcal{A}^-}\leq 1-\mathbb{E}X_{b,\mathcal{A}^-}
\nonumber\\
\beta_m-1-c_{\{a,b\},\mathcal{A}^{--}}&\leq\bar{X}_{\{a,b\},\mathcal{A}^{--}}\leq 1-\beta_m-c_{\{a,b\},\mathcal{A}^{--}},
\end{align}
where $c_{\{a,b\},\mathcal{A}^{--}}=\mathbb{E}\left(X_{a,\{b,\mathcal{A}^{--}\}}^{(\ell)}\right)
-\mathbb{E}\left(X_{b,\{a,\mathcal{A}^{--}\}}^{(\ell)}\right)$.

With these bounds, we can bound the denominator in the right hand side of (\ref{ineq: Hoeffding}) as follows
\begin{align}  
    &\sum_{\ell\in{[r]}}\sum_{\substack{\mathcal{A}^-\subseteq{[n]}\setminus\left\{a,b\right\}\\\left|\mathcal{A}^-\right|=m-1}}1+\sum_{\ell\in{[r]}}\sum_{\substack{\mathcal{A}^-\subseteq{[n]}\setminus\left\{a,b\right\}\\\left|\mathcal{A}^-\right|=m-1}}1\notag\\
    &\qquad+\sum_{\ell\in{[r]}}\sum_{\substack{\mathcal{A}^{--}\subseteq{[n]}\setminus\left\{a,b\right\}\\\left|\mathcal{A}^{--}\right|=m-2}}(2(1-\beta_m))^2\notag\\
    &=r\rho_{n,m}(4(1-\beta_m)^2(\frac{m-1}{n-1})+2\frac{n-m}{n-1}).
\end{align}
Applying the Hoeffding's inequality, we obtain
\begin{align}
&\mathbf{Pr}(W_b-W_a> 0)\nonumber\\
&\leq\exp\left(\frac{-\left(rp\rho_{n,m}\Delta_k\right)^2}{r\rho_{n,m}(2(1-\beta_m)^2\frac{m-1}{n-1}+\frac{n-m}{n-1})}\right)\notag\\
&\leq n^{-\frac{\alpha^2p}{(2(1-\beta_m)^2(\frac{m-1}{n-1})+\frac{n-m}{n-1})}}
\end{align}
Taking a union bound over the possible pairs $(a,b)$ gives the desired result.
\end{proof}

\begin{proof}[Proof of Proposition \ref{lemma:forward1}]
The proof of Proposition \ref{lemma:forward1} is similar to the Proposition \ref{lemma:forward2}, and the difference is in the way of bounding the probability on the right-hand side of (\ref{ineq:probability of separated by delta}).
Instead of Hoeffding's inequality, we invoke Bernstein's inequality, which is stated below for completeness. 
\begin{lemma}[Bernstein's inequality]\label{lemma:Bernstein}
Let $Y_{1},...,Y_{n}$ be independent zero-mean random variables. Suppose that  $|Y_{i}|\leq M$  almost surely, for all  $i$. Then, for all positive $t$,
 we have that
 $$ \mathbb {P} \left(\sum _{i=1}^{n}Y_{i}>t\right)\leq \exp \left(-{\frac {{\tfrac {1}{2}}t^{2}}{\sum \mathbb {E} \left[Y_{i}^{2}\right]+{\tfrac {1}{3}}Mt}}\right).$$
 \end{lemma}
 Since the centralized scores and centralized cross-scores are zero-mean and independent, we only need to bound the sum of the variances of the centralized random variables, for which we have the following lemma.

\begin{lemma}
\label{lemma:variancebound}
\begin{align}
&\Sigma_{\bar{X}_{a,b}}\triangleq\sum_{\substack{\mathcal{A}^-\subseteq[n]\setminus\{a,b\}}}\left(\mathbb{E}\left(\bar{X}_{b,\mathcal{A}^-}^{(\ell)}\right)^2+\mathbb{E}\left(\bar{X}_{a,\mathcal{A}^-}^{(\ell)}\right)^2\right)\nonumber\\
&\qquad\qquad\qquad+\sum_{\substack{\mathcal{A}^{--}\subseteq [n]\setminus\{a,b\}}}\mathbb{E}\left(\bar{X}_{\{a,b\},\mathcal{A}^{--}}^{(\ell)}\right)^2\nonumber\\
&\qquad\leq (2p-p^2)\rho_{n,m}-p\rho_{n,m}\Delta_k.
\end{align}
\end{lemma}

The proof of this lemma can be found in the appendix. With this lemma, we can apply Bernstein's inequality on $\mathbf{Pr}(W_b-W_a>0)$, from which we obtain
\begin{align}
&\mathbf{Pr}(W_b-W_a> 0)\notag\\
&\leq \exp\left(\frac{-\frac{1}{2}\left(rp\rho_{n,m}\Delta_k\right)^2}{ rp(2-p)\rho_{n,m}-rp\rho_{n,m}\Delta_k+\frac{2}{3}\rho_{n,m}rp\Delta_k}\right)\notag\\
&\leq n^{-\frac{\alpha^2}{(4-2p)}}.
\end{align}
Taking the union bound over the possible pairs $(a,b)$ gives the desired result.
\end{proof}

\section{Proof of Theorem \ref{theorem:converse}}
\label{sec:proofconverse}

For each $a\in\{k,\ldots,n\}$, denote the $k$-sized subset $S^*[a]=\{1,2,\ldots,k-1\}\cup\{a\}$. In order to prove the converse result, we construct a probability distribution $M\in \mathscr{F}_k(\alpha)$ such that it is difficult for any method to find the top-$k$ items, which is:
\begin{align}
M^a_{\vec{v}}=\left\{
\begin{array}{ll}
(m!)^{-1}(1+\delta)& \text{if C1}\\
(m!)^{-1}(1-\delta)& \text{if C2}\\
(m!)^{-1}& \text{otherwise}
\end{array}
\right., \label{eqn:worsedist}
\end{align}
where $\delta\in [0,1]$ is a parameter to be specified later, and the conditions C1 and C2 are given as: 
\begin{itemize}
    \item C1: $\vec{v}_{1:q}\subseteq\mathcal{S}^*[a] \text{ and }\vec{v}_{q+1:m}\subseteq [n]\setminus\mathcal{S}^*[a]$;
    \item C2: $\vec{v}_{1:m-q}\subseteq [n]\setminus\mathcal{S}^*[a] \text{ and } \vec{v}_{m-q+1:m}\subseteq\mathcal{S}^*[a]$.
\end{itemize}

The theorem is established through a lemma and a proposition. 

\begin{lemma}
\label{prop:converse1}
For the distribution $M^a_{\vec{v}}$, where $a\in\{k,\ldots,n\}$, we have 
\begin{align}
\Delta_k=\frac{g(n,m,\vec{\beta})}{\rho_{n,m}}\delta. \label{eqn:Deltak}
\end{align}
\end{lemma}

The proof of this lemma can be found in the appendix. 

\begin{prop}
\label{prop:converse2}
For any distinct $a,b\in \{k,k+1,\ldots,n\}$, 
\begin{align}
D_{\text{KL}}(\mathbb{P}^a||\mathbb{P}^b) \leq rph(n,m)\frac{4\delta^2}{(1-\delta^2)}.
\end{align}
\end{prop}

We are now ready to prove Theorem \ref{theorem:converse}. Suppose the underlying distribution is drawn uniformly at random from the set $\left\{M^a| a\in [n]\setminus[k-1]\right\} $, and the true index is $a^*$. By Fano's inequality, any estimator $\hat{a}$ must have an error probability lower-bounded by
\begin{align}
\label{eqn:fano}
\mathbb{P}_M[\hat{a}\neq a^*]\geq 1-\frac{rph(n,m)\frac{4\delta^2}{(1-\delta^2)}+\log 2}{\log (n-k+1)}.
\end{align}
Since $\alpha\leq\frac{\sqrt{2}}{7}g(n,m,\vec{\beta})\sqrt{\frac{1}{h(n,m)\rho_{n,m}}}$ and  $p\geq\frac{\log n}{4rh(n,m)}$, we can choose $\Delta_k=\frac{\sqrt{2}}{7}g(n,m,\vec{\beta})\sqrt{\frac{1}{h(n,m)\rho_{n,m}}}\sqrt{\frac{\log n}{rp\rho_{n,m}}}$. Combining these bounds with (\ref{eqn:Deltak}) in  Proposition \ref{prop:converse1} gives
\begin{equation}
\label{eqn:delta_upperbound}
\delta\leq\frac{\sqrt{2}}{7}\sqrt{\frac{\log n}{rph(n,m)}}\leq\frac{2\sqrt{2}}{7}.
\end{equation}
Finally (\ref{eqn:fano}) and (\ref{eqn:delta_upperbound}) give
\begin{align}
\mathbb{P}_M[\hat{a}\neq a^*]&\geq1-\frac{rph(n,m)\frac{4{(\frac{\sqrt{2}}{7}\sqrt{\frac{\log n}{rph(n,m)}})}^2}{(1-{(\frac{2\sqrt{2}}{7})}^2)}+\log 2}{\log n-\log 2}\nonumber\\
&\geq 1-\frac{\frac{8}{41}\log n+\log 2}{\log n-\log 2}\geq\frac{1}{7},
\end{align}
where we also used $2k\leq n$ and $n\geq 7$, and this completes the proof.

\section{Conclusion}

We have considered using the Borda counting algorithm on noisy $m$-wise ranking data to select the top-$k$ items, which is a generalization of a previous work using only pairwise comparisons. Our analysis confirmed the importance of  the associated score separation $\Delta_k$. In high and low observation probability regimes, our analysis revealed that the error probability of the algorithm can be bounded differently. The resultant bound is also tighter than that given by Shah et al. when we apply the analysis back to the pairwise case. For the converse direction, we established an error probability lower bound. The gap between the upper and lower bound was analyzed. These results were further extended to the approximate top-$k$ selection problem. Through numerical simulation, we have observed that the Borda counting-based algorithm is competitive to the spectral MLE-based algorithm, particularly in the high observation probability regime.

\appendix
In this section, we first establish a relation between the PL model and the non-parametric model. Recall the definition of the PL model:
\begin{definition}
\label{def:PL}
\textbf{Plackett-Luce model.} The PL model is a parametric model assuming the existence of a ``quality" parameter $w_i\in\mathbb{R}$ for each item $i$, and requiring that the probability of observing the ranking $\vec{v}=(v_1,v_2,...,v_m)$ for the ranked comparison among $\vec{v}\doteq\mathcal{A}$ is given by a specific function. More precisely, 
    \begin{equation}
        M_{v_1,v_2,...,v_m}=M_{\vec{v}}=\prod_{k=1}^m f_k(\vec{v}),
    \end{equation}
    where 
    \begin{equation}
        f_k(\vec{v})=\frac{w_{v_k}}{\sum_{i=k}^nw_{v_i}}.\label{eqn:pl}
    \end{equation}
\end{definition}
%The ranking distribution $f_k(\vec{v})$ can also be induced by the score distribution generated by a Thurstonian model\cite{YELLOTT1977109}. To be specific, the ranking of the items is determined by their unobserved scores. For an item $i$, its score $x_i$ follows the Gumbel distribution PDF

%Where $z(x) = e^{-\frac{x-\mu}{\beta}}$. If $\beta$ is fixed and $e^{\frac{\mu_i}{\beta}}=w_i$ for each item $i$, then the distribution of the scores' ranking will reduce to equation (\ref{eqn:pl}).
The PL model can also be recovered from the Thurstone model \cite{hajek2014minimax}.
\begin{definition}[\cite{hajek2014minimax}] 
\label{def:PL_Thurstone}
A partial ranking $\sigma$ : $[|S|] \rightarrow S$
is generated from $\{\theta^*_i
, i \in S\}$ under the PL model in two steps:  (1) independently assign each item $i\in S$ an unobserved value $Z_i$
,
exponentially distributed with mean $e^{-\theta^*_i}$ ;
(2) select $\sigma$ so that $Z_{\sigma(1)} \leq Z_{\sigma(2)} \leq...\leq Z_{\sigma(|S|)}$
. Here $\theta^*_i=-\log w_i$.
\end{definition}
\begin{lemma}[Restatement of Lemma 1]
If the underlying distribution follows the PL model, then the ranking of the associate scores is consistent with the ranking of the weight vector regardless of the assignment of score vector $\vec{\beta}$, i.e., $\forall\vec{\beta}$ satisfying $1=\beta_{1}\geq \beta_{2}\geq ...\geq \beta_{m}\geq 0$, and $\forall a, b\in[n]$, if $w_a\geq w_b$, we have $\tau_a\geq\tau_b$.
\end{lemma}
\begin{proof}
By the definition of associate score, we have
\begin{align}
    \tau_a-\tau_b&=\frac{1}{\rho_{n,m}}\bigg{(}\sum_{\substack{\mathcal{A}^-\subseteq [n]\setminus\{a\}}}\sum_{t=1}^{m}\beta_{t}R_{a,\mathcal{A}^-}(t)\bigg{)}\notag\\
    &\qquad-\frac{1}{\rho_{n,m}}\bigg{(}\sum_{\substack{\mathcal{A}^-\subseteq [n]\setminus\{b\}}}\sum_{t=1}^{m}\beta_{t}R_{b,\mathcal{A}^-}(t)\bigg{)}\notag\\
    &=\frac{1}{\rho_{n,m}}\bigg{(}\sum_{\substack{\mathcal{A}\subseteq [n]\\a,b\in\cA}}(\sum_{i=1}^{m}\sumV{v_i=a}\beta_{i}M_{\vec{v}}-\sum_{j=1}^{m}\sumV{v_j=b}\beta_{j}M_{\vec{v}})\bigg{)}\notag\\
    &\qquad+\frac{1}{\rho_{n,m}}\bigg{(}\sum_{\substack{\mathcal{A}\subseteq [n]\\a\in\cA,b\notin\cA}}\sum_{i=1}^{m}\sumV{v_i=a}\beta_{i}M_{\vec{v}}\notag\\
    &\qquad-\sum_{\substack{\mathcal{A}\subseteq [n]\\b\in\cA,a\notin\cA}}\sum_{j=1}^{m}\sumV{v_j=b}\beta_{j}M_{\vec{v}}\bigg{)}.
\end{align}
First, we prove $\frac{1}{\rho_{n,m}}\bigg{(}\sum_{\substack{\mathcal{A}\subseteq [n]\\a,b\in\cA}}(\sum_{i=1}^{m}\sumV{v_i=a}\beta_{i}M_{\vec{v}}-\sum_{j=1}^{m}\sumV{v_j=b}\beta_{j}M_{\vec{v}})\bigg{)}\geq0$. For any $\mathcal{A}\subseteq [n]$ satisfying $a,b\in\cA$, we have
\begin{align}
    &\sum_{i=1}^{m}\sumV{v_i=a}\beta_{i}M_{\vec{v}}-\sum_{j=1}^{m}\sumV{v_j=b}\beta_{j}M_{\vec{v}}\notag\\
    =&\sum_{i,j=1}^{m}\sumV{v_i=a\\v_j=b}(\beta_{i}-\beta_j)M_{\vec{v}}\notag\\
    =&\sum_{\substack{i,j=1\\i<j}}^{m}\sumV{v_i=a\\v_j=b}(\beta_{i}-\beta_j)M_{\vec{v}}
   % \notag\\
    %&\qquad
    +\sum_{\substack{i,j=1\\i<j}}^{m}\sumV{v_i=b\\v_j=a}(\beta_{j}-\beta_i)M_{\vec{v}}\notag\\
    =&\sum_{\substack{i,j=1\\i<j}}^{m}\sum_{(\vec{v}_1,\vec{v}_2)\in\cC_{i,j}}(\beta_{i}-\beta_j)(M_{\vec{v}_1}-M_{\vec{v}_2}),
\end{align}
where $\cC_{i,j}:=\{(\vec{v}_1,\vec{v}_2)\in[n]^{2m}\text{ : } v_1(a)=i,  v_1(b)=j, v_2(a)=j,v_2(b)=i\text{, and } v_1(x)=v_2(x),\forall x\in\mathcal{A}\backslash\{a,b\}\}$.
\begin{align}
    &M_{\vec{v}_1}-M_{\vec{v}_2}=\M{1}-\M{2}\notag\\
    &=(\fprodT{k=1,k\neq i,j})\{f_i(\vec{v}_1)f_j(\vec{v}_1)-f_i(\vec{v}_2)f_j(\vec{v}_2)\}\notag\\
    &\qquad+(\fprod{k\leq i,k\geq j})\{\fprod{i<k<j}-\fprodT{i<k<j}\},
\end{align}
where the second equality follows from the fact that $f_k(\vec{v}_1)=f_k(\vec{v}_2)$, $\forall k<i$ and $\forall k>j$.
\begin{align}
    &f_i(\vec{v}_1)f_j(\vec{v}_1)-f_i(\vec{v}_2)f_j(\vec{v}_2)\notag\\
    =&\frac{w_a}{w_a+\sum_{t=i+1}^mw_{v_t}}\frac{w_b}{w_b+\sum_{t=j+1}^mw_{v_t}}\notag\\
    &\qquad-\frac{w_b}{w_b+\sum_{t=i+1}^mw_{v_t}}\frac{w_a}{w_a+\sum_{t=j+1}^mw_{v_t}}\geq0
\end{align}
Since $f_k(\vec{v}_1)\geq f_k(\vec{v}_2)$, for $i<k<j$, we have $M_{\vec{v}_1}-M_{\vec{v}_2}\geq0$, and thus $\frac{1}{\rho_{n,m}}\bigg{(}\sum_{\substack{\mathcal{A}\subseteq [n]\\a,b\in\cA}}(\sum_{i=1}^{m}\sumV{v_i=a}\beta_{i}M_{\vec{v}}-\sum_{j=1}^{m}\sumV{v_j=b}\beta_{j}M_{\vec{v}})\bigg{)}\geq0$.

Second, we need to prove $$\frac{1}{\rho_{n,m}}\bigg{(}\sum_{\substack{\mathcal{A}\subseteq [n]\\a\in\cA,b\notin\cA}}\sum_{i=1}^{m}\sumV{v_i=a}\beta_{i}M_{\vec{v}}-\sum_{\substack{\mathcal{A}\subseteq [n]\\b\in\cA,a\notin\cA}}\sum_{j=1}^{m}\sumV{v_j=b}\beta_{j}M_{\vec{v}}\bigg{)}\geq0$$
This is equivalent to proving that $\sum_{i=1}^{m}\sumV{v_i=a}\beta_{i}M_{\vec{v}}$ is a non-decreasing function in $w_a$. We denote that $\xi_{\mathcal{A}^-}(w_a):=\sum_{i=1}^{m}\sumV{v_i=a}\beta_{i}M_{\vec{v}}$. From the definition of $X_{a,\mathcal{A}^-}$, we have $\xi_{\mathcal{A}^-}(w_a)=\mathbb{E}X_{a,\mathcal{A}^-}$. Define $\vec{\mathcal{A}}$ as an ordered instance of $\mathcal{A}$, i.e., $\vec{\mathcal{A}}\doteq\mathcal{A}$. We also define the mapping $\phi_a:\mathbb{R}_+^{m}\rightarrow [0,1]$. Here $\phi_a(z_{\vec{\mathcal{A}}})$ represents the score item $a$ receives when the sampled results in Definition \ref{def:PL_Thurstone} is $z_{\vec{\mathcal{A}}}$. From the Thurstonian representation of the PL model, we have 
\begin{align*}
    \mathbb{E}X_{a,\mathcal{A}^-}&=\int_{z_{\vec{\mathcal{A}}}\in\mathbb{R}_+^{m}}\phi_a(z_{\vec{\mathcal{A}}})dF(z_{\vec{\mathcal{A}}})\\
    &=\int_{z_{\vec{\mathcal{A}}}\in\mathbb{R}_+^{m}}\phi_a(z_{\vec{\mathcal{A}}^-},z)\frac{1}{w_a}e^{-\frac{z}{w_a}}dF(z_{\vec{\mathcal{A}}^-})dz\\
    &=\int_{z_{\vec{\mathcal{A}}}\in\mathbb{R}_+^{m}}\phi_a(z_{\vec{\mathcal{A}}^-},w_ay)e^{-y}dF(z_{\vec{\mathcal{A}}^-})dy.
\end{align*}
For fixed $y$, if $w_a$ increases, $w_ay$ increases. This implies that the ranking result of item $a$ will be non-decreasing, thus $\phi_a(z_{\vec{\mathcal{A}}^-},w_ay)$ is non-decreasing in $w_a$, we then have  $\mathbb{E}X_{a,\mathcal{A}^-}$ is non-decreasing in $w_a$, which concludes the proof.
\end{proof}

We next provide the technical proofs for several auxiliary lemmas and that for Proposition \ref{prop:converse2}. 

\begin{proof}[Proof of Lemma \ref{lemma:variancebound}]
We first calculate the variances of the centralized random variables as follows. For any $\mathcal{A}^-\in{[n]}\backslash\left\{a,b\right\}$ and $\left|\mathcal{A}^-\right|=m-1$, we have for the centralized scores
\begin{align}
\mathbb{E}\left[\left(\bar{X}_{a,\mathcal{A}^-}^{(\ell)}\right)^2\right]
&=\mathbb{E}\left[\left(X_{a,\mathcal{A}^-}^{(\ell)}\right)^2\right]-\mathbb{E}\left[X_{a,\mathcal{A}^-}^{(\ell)}\right]^2\nonumber\\
&=p\sum_{t=1}^{m}\beta_t^2R_{a,\mathcal{A}^-}(t)-(p\sum_{t=1}^{m}\beta_tR_{a,\mathcal{A}^-}(t))^2\nonumber\\
\label{eqn:bound of variance}
&\leq p\sum_{t=1}^{m}\beta_tR_{a,\mathcal{A}^-}(t)-(p\sum_{t=1}^{m}\beta_tR_{a,\mathcal{A}^-}(t))^2.
\end{align}
The inequality in (\ref{eqn:bound of variance}) holds because for any $i\in[m]$, we have $\beta_i\leq 1$, and thus $\beta_i^2\leq\beta_i$.
Then similarly, for variance of centralized score related to item $b$,
\begin{align}
\mathbb{E}\left[\left(\bar{X}_{b,\mathcal{A}^-}^{(\ell)}\right)^2\right]
&\leq p\sum_{t=1}^{m}\beta_tR_{b,\mathcal{A}^-}(t)
\end{align}
  
Moreover, we have for any $a,b,\mathcal{A}^{--}$ 
\begin{align}
&\mathbb{E}\left[\left(\bar{X}_{\{a,b\},\mathcal{A}^{--}}^{(\ell)}\right)^2\right]\nonumber\\
&= \mathbb{E}\left[\left(X^{(\ell)}_{a,\{b,\mathcal{A}^{--}\}}-X^{(\ell)}_{b,\{a,\mathcal{A}^{--}\}}\right)^2\right]\notag\\
&\qquad-\left(\mathbb{E}X^{(\ell)}_{a,\{b,\mathcal{A}^{--}\}}-\mathbb{E}X^{(\ell)}_{b,\{a,\mathcal{A}^{--}\}}\right)^2\nonumber\\
%&\leq p-\left(\mathbb{E}X^{(\ell)}_{a,\{b,\mathcal{A}^{--}\}}-\mathbb{E}X^{(\ell)}_{b,\{a,\mathcal{A}^{--}\}}\right)^2\nonumber\\
%&= (p\sum_{t=1}^m\beta_{t}^2 R^t_{a,\{b,\mathcal{A}^{--}\}}-p\sum_{t=1}^m\beta_{t}^2R^t_{b,\{a,\mathcal{A}^{--}\}})-(p\sum_{t=1}^m\beta_{t} R^t_{a,\{b,\mathcal{A}^{--}\}}-p\sum_{t=1}^m\beta_{t}R^t_{b,\{a,\mathcal{A}^{--}\}})^2\nonumber\\
&\leq p-(p\sum_{t=1}^m\beta_{t} R^t_{a,\{b,\mathcal{A}^{--}\}}-p\sum_{t=1}^m\beta_{t}R^t_{b,\{a,\mathcal{A}^{--}\}})^2.
\end{align}
\\

%\begin{align}
%&\mathbb{E}\left[\left(\bar{X}_{\{a,b\},\mathcal{A}^{--}}^{(\ell)}\right)^2\right]\nonumber\\
%&\leq \mathbb{E}\left[\left(X^{(\ell)}_{a,\{b,\mathcal{A}^{--}\}}-X^{(\ell)}_{b,\{a,\mathcal{A}^{--}\}}\right)^2\right]\nonumber\\
%&=p\sum_{j,k\in{[m]}}(\beta_{j}-\beta_{k})^2\sum_{\substack{\vec{v}\doteq \mathcal{A}^{--}\cup\{a,b\}\\v_j=a,v_k=b}}M_{\vec{v}}\nonumber\\
%&\leq p\sum_{j,k\in{[m]}}(\beta^2_{j}+\beta^2_{k})\sum_{\substack{\vec{v}\doteq \mathcal{A}^{--}\cup\{a,b\}\\v_j=a,v_k=b}}M_{\vec{v}}\nonumber\\
%&= p\sum_{j=1}^m\beta_{j}^2\sum_{\substack{\vec{v}\doteq \mathcal{A}^{--}\cup\{a,b\}\\v_j=a}}M_{\vec{v}}\\
%&\qquad\qquad+p\sum_{k=1}^m\beta_{k}^2\sum_{\substack{\vec{v}\doteq \mathcal{A}^{--}\cup\{a,b\}\\v_k=b}}M_{\vec{v}}\nonumber\\
%&= p\sum_{t=1}^m\beta_{t}^2\left[ R_{a,\{b,\mathcal{A}^{--}\}}(t)+R_{b,\{a,\mathcal{A}^{--}\}}(t)\right].
%\end{align}

Thus we have the following bound on the sum of the variances of centralized scores and the cross-scores
\begin{align}
\label{eqn:var_bound}
&\Sigma_{\bar{X}_{a,b}}\triangleq\sum_{\substack{\mathcal{A}^-\subseteq[n]\setminus\{a,b\}\\|\mathcal{A}^-|=m-1}}\left(\mathbb{E}\left(\bar{X}_{b,\mathcal{A}^-}^{(\ell)}\right)^2+\mathbb{E}\left(\bar{X}_{a,\mathcal{A}^-}^{(\ell)}\right)^2\right)\notag\\
&\qquad+\sum_{\substack{\mathcal{A}^{--}\subseteq [n]\setminus\{a,b\}\\|\mathcal{A}^{--}|=m-2}}\mathbb{E}\left(\bar{X}_{\{a,b\},\mathcal{A}^{--}}^{(\ell)}\right)^2\nonumber\\
&\leq p\sum_{\substack{\mathcal{A}^-\subseteq{[n]}\setminus\left\{a,b\right\}\\\left|\mathcal{A}^-\right|=m-1}}\sum_{t=1}^{m}\beta_{t}R_{a,\mathcal{A}^-}(t)\nonumber\\
&\qquad-\sum_{\substack{\mathcal{A}^-\subseteq{[n]}\setminus\left\{a,b\right\}\\\left|\mathcal{A}^-\right|=m-1}}(p\sum_{t=1}^{m}\beta_t R_{a,\mathcal{A}^-}(t))^2\notag\\
&\qquad+p\sum_{\substack{\mathcal{A}^-\subseteq{[n]}\setminus\left\{a,b\right\}\\\left|\mathcal{A}^-\right|=m-1}}\sum_{t=1}^{m}\beta_{t}R_{b,\mathcal{A}^-}(t) +\sum_{\substack{\mathcal{A}^{--}\subseteq{[n]}\setminus\left\{a,b\right\}\\\left|\mathcal{A}^{--}\right|=m-2}}p\nonumber\\
&- \sum_{\substack{\mathcal{A}^{--}\subseteq{[n]}\setminus\left\{a,b\right\}\\\left|\mathcal{A}^{--}\right|=m-2}}p^2(\sum_{t=1}^m\beta_{t} (R_{a,\{b,\mathcal{A}^{--}\}}(t)-R_{b,\{a,\mathcal{A}^{--}\}}(t)))^2.
\end{align}
Define $S^{(i)}_{a,b\mathcal{A}^{--}}=p(\sum_{t=1}^m\beta_{t} (R_{a,\{b,\mathcal{A}^{--}\}}(t)-R_{b,\{a,\mathcal{A}^{--}\}}(t)))$ for notation simplicity. By the definition of associate score, we get
\begin{align}
    &\Sigma_{\bar{X}_{a,b}}\leq p\rho_{n,m}(\tau_a+\tau_b)-\sum_{\substack{\mathcal{A}^-\subseteq{[n]}\setminus\left\{a,b\right\}\\\left|\mathcal{A}^-\right|=m-1}}(p\sum_{t=1}^{m}\beta_t R_{a,\mathcal{A}^-}(t))^2\notag\\
&+ \sum_{\substack{\mathcal{A}^{--}\subseteq{[n]}\setminus\left\{a,b\right\}\\\left|\mathcal{A}^{--}\right|=m-2}}(S^{(i)}_{a,b\mathcal{A}^{--}}-(S^{(i)}_{a,b\mathcal{A}^{--}})^2)\notag\\
&-2p\sum_{\substack{\mathcal{A}^{--}\subseteq{[n]}\setminus\left\{a,b\right\}\\\left|\mathcal{A}^{--}\right|=m-2}}\sum_{t=1}^{m}\beta_{t}R
_{a,b,\mathcal{A}^{--}}(t)+p\binom{n-2}{m-2}.
\end{align}
By the fact that $\tau_a-\tau_b\geq\Delta_k$, we can get
\begin{align}
&\Sigma_{\bar{X}_{a,b}}\nonumber\\
&\leq 2p
\sum_{\substack{\mathcal{A}^-\subseteq{[n]}\setminus\left\{a,b\right\}\\\left|\mathcal{A}^-\right|=m-1}}\sum_{t=1}^{m}\beta_{t}R_{a,\mathcal{A}^-}(t)\notag\\
&-\sum_{\substack{\mathcal{A}^-\subseteq{[n]}\setminus\left\{a,b\right\}\\\left|\mathcal{A}^-\right|=m-1}}(p\sum_{t=1}^{m}\beta_tR_{a,\mathcal{A}^-}(t))^2-p\rho_{n,m}\Delta_k\nonumber\\
&+\sum_{\substack{\mathcal{A}^{--}\subseteq{[n]}\setminus\left\{a,b\right\}\\\left|\mathcal{A}^{--}\right|=m-2}}(S^{(i)}_{a,b\mathcal{A}^{--}}-(S^{(i)}_{a,b\mathcal{A}^{--}})^2)+p\binom{n-2}{m-2}.
%&\leq \sum_{\substack{\mathcal{A}^-\subseteq{[n]}\setminus\left\{a,b\right\}\\\left|\mathcal{A}^-\right|=m-1}}(2p\sum_{t=1}^{m}\beta_{t}R_{a,\mathcal{A}^-}(t)\nonumber\\
%&\qquad-(p\sum_{t=1}^{m}\beta_tR_{a,\mathcal{A}^-}(t))^2)-p\rho_{n,m}\Delta_k+p\binom{n-2}{m-2}\nonumber\\
%&+ \sum_{\substack{\mathcal{A}^{--}\subseteq{[n]}\setminus\left\{a,b\right\}\\\left|\mathcal{A}^{--}\right|=m-2}}\left[(p\sum_{t=1}^{m}\beta_{t}R
%_{a,b,\mathcal{A}^{--}}(t)-p\sum_{t=1}^{m}\beta_{t}R_{b,a,\mathcal{A}^{--}}(t))\right]\nonumber\\
%&-\sum_{\substack{\mathcal{A}^{--}\subseteq{[n]}\setminus\left\{a,b\right\}\\\left|\mathcal{A}^{--}\right|=m-2}}\left[(p\sum_{t=1}^m\beta_{t} R_{a,\{b,\mathcal{A}^{--}\}}(t)+p\sum_{t=1}^m\beta_{t}R_{b,\{a,\mathcal{A}^{--}\}}(t))^2\right] .
\end{align}

Since  $ \sum_{t=1}^{m}\beta_tR_{i,\mathcal{A}^-}(t)\leq\sum_{t=1}^{m}R_{i,\mathcal{A}^-}(t)=1$
holds for both $i=a$ and $i=b$, we further have 
\begin{align}
  -p\leq\left(p\sum_{t=1}^{m}\beta_{t}R
_{a,b,\mathcal{A}^{--}}(t)-p\sum_{t=1}^{m}\beta_{t}R_{b,a,\mathcal{A}^{--}}(t)\right)\leq p.
\end{align}
 According to the assumption in Lemma \ref{lemma:forward1}, $p\leq\frac{1}{2}$ we can get that
\begin{align}
    S^{(i)}_{a,b\mathcal{A}^{--}}-(S^{(i)}_{a,b\mathcal{A}^{--}})^2\leq p-p^2.
\end{align}
Similarly, we can obtain that
\begin{equation}
    2p\sum_{t=1}^{m}\beta_{t}R_{a,\mathcal{A}^-}(t)-(p\sum_{t=1}^{m}\beta_tR_{a,\mathcal{A}^-}(t))^2\leq 2p-p^2.
\end{equation}
It follows that the total variance can be bounded as
\begin{align}
\Sigma_{\bar{X}_{a,b}}&\leq\left[\sum_{\substack{\mathcal{A}^-\subseteq{[n]}\setminus\left\{a,b\right\}\\\left|\mathcal{A}^-\right|=m-1}}(2p-p^2)\right]-p\rho_{n,m}\Delta_k+p\binom{n-2}{m-2}\nonumber\\
&\qquad+\left[ \sum_{\substack{\mathcal{A}^{--}\subseteq{[n]}\setminus\left\{a,b\right\}\\\left|\mathcal{A}^{--}\right|=m-2}}(p-p^2)\right]\nonumber\\
&=(2p-p^2)\rho_{n,m}-p\rho_{n,m}\Delta_k.
\end{align}
This is the desired result, and the proof is complete. 
\end{proof}

\begin{proof}[Proof of Lemma \ref{prop:converse1}]
Let us consider the associated score of item-$k$ using the distribution $M^*$, which can be seen as
\begin{align}
\tau_{(k)}&=\frac{1}{m\binom{n-1}{m-1}}{n \choose m-1}\sum_{t=1}^{m}\beta_{t}\notag\\
&\quad+\frac{\delta}{\binom{n-1}{m-1}}\frac{1}{m!}\sum_{t=1}^{q}(\beta_{t}-\beta_{m-q+t})A^{q-1}_{k-1}A_{n-k}^{m-q}
%\\
%&+\frac{\delta}{f(n)}\sum_{m=2}^{n-k+1}p(\beta_{m,1}-\beta_{m,m})A_{n-k}^{m-1}\nonumber\\
%&\qquad\qquad+\frac{1}{f(n)}\sum_{m=n-k+2}^{n}p\sum_{t=1}^{m-n+k}(\beta_{m,t}-\beta_{m,m+1-t})A^{m-n+k}_{k},
\end{align}
Similarly, we have 
\begin{align}
\tau_{(k+1)}&=\frac{1}{\binom{n-1}{m-1}}\frac{1}{m}{n \choose m-1}\sum_{t=1}^{m}\beta_{t}\notag\\
&\quad+\frac{\delta}{\binom{n-1}{m-1}}\frac{1}{m!}\sum_{t=q+1}^{m}(\beta_{t}-\beta_{t-q})A^{q}_{k}A_{n-k-1}^{m-1-q}.%\\
%&+\frac{\delta}{\binom{n-1}{m-1}}\sum_{m=2}^{n-k+1}p_m\sum_{t=2}^{m}(\beta_{m,t}-\beta_{m,m+1-t})kA_{n-k-1}^{m-2}\nonumber\\
%&+\frac{\delta}{\binom{n-1}{m-1}}\sum_{m=n-k+2}^{n}p_m\sum_{t=q+1}^{m}(\beta_{m,t}-\beta_{m,m+1-t})A^{m-n+k}_{k}A^{n-k-1}_{n-k-1},
\end{align}
It follows that
\begin{align}
\Delta_k=\tau_{(k)}-\tau_{(k+1)},
\end{align}
which leads to the quantity in (\ref{eqn:Deltak}), after the first terms in both $\tau_{(k)}$ and $\tau_{(k+1)}$ cancel out each other.
\end{proof}

\begin{proof}[Proof of Proposition \ref{prop:converse2}]
For notation convenience, we shall also define
\begin{align*}
&\mu(\delta)\triangleq p\left[(1+\delta)\log (1+\delta)+(1-\delta)\log (1-\delta)\right]\nonumber\\
&\upsilon(\delta)\triangleq p\left[\log (1+\delta)^{-1}+\log (1-\delta)^{-1}\right],\\
&\omega_{m,q}\triangleq {m\choose q}^{-1}.
\end{align*}

Using the property of the KL-divergence, we can write
\begin{align}
&D_{\text{KL}}(\mathbb{P}^a||\mathbb{P}^b)=\sum_{\ell\in[r]} p \sum_{\mathcal{A}\subseteq [n]:|\mathcal{A}|=m}D_{\text{KL}}(\mathbb{P}^a(V_{\mathcal{A}}^{(\ell)})||\mathbb{P}^b(V_{\mathcal{A}}^{(\ell)}))\nonumber\\
&\qquad\qquad=r p \sum_{\mathcal{A}\subseteq [n]:|\mathcal{A}|=m} D_{\text{KL}}(\mathbb{P}^a(V_{\mathcal{A}}^{(\ell)})||\mathbb{P}^b(V_{\mathcal{A}}^{(\ell)})),
\end{align}
where $V_{\mathcal{A}}^{(\ell)}$ is the result of the comparison among the elements in $\mathcal{A}$. It is clear that we need to bound the KL divergence for each subset $\mathcal{A}\subseteq[n]$, which can be done as follows.

\begin{enumerate}
\item If $a,b\notin \mathcal{A}$, then clearly $D_{\text{KL}}(\mathbb{P}^a(V_{\mathcal{A}}^{(\ell)})||\mathbb{P}^b(V_{\mathcal{A}}^{(\ell)}))=0$;
\item If $a\in \mathcal{A}$ but $b\notin \mathcal{A}$, 
\begin{enumerate}
\item If $\mathcal{A}\cap[k-1]=q-1$, we have
\begin{align}
&D_{\text{KL}}(\mathbb{P}^a(V_{\mathcal{A}}^{(\ell)})||\mathbb{P}^b(V_{\mathcal{A}}^{(\ell)}))=\frac{q!(m-q)!}{(m!)}\mu(\delta)\nonumber\\
&\qquad\qquad\leq 2\omega_{m,q}p\delta^2\leq 2\omega_{m,q}\frac{p\delta^2}{1-\delta^2}.
\end{align}
\item If $\mathcal{A}\cap[k-1]=q$, we have
\begin{align}
&D_{\text{KL}}(\mathbb{P}^a(V_{\mathcal{A}}^{(\ell)})||\mathbb{P}^b(V_{\mathcal{A}}^{(\ell)}))=\frac{q!(m-q)!}{m!}\upsilon(\delta)\nonumber\\
&\leq \omega_{m,q}p\left(\frac{-\delta}{1+\delta}+\frac{\delta}{1-\delta}\right)=2\omega_{m,q}p\frac{\delta^2}{1-\delta^2}.
\end{align}
\item If $\mathcal{A}\cap[k-1]\neq q-1$ and $\mathcal{A}\cap[k-1]\neq q$, then
\begin{equation}
D_{\text{KL}}(\mathbb{P}^a(V_{\mathcal{A}}^{(\ell)})||\mathbb{P}^b(V_{\mathcal{A}}^{(\ell)}))=0. 
\end{equation}
\end{enumerate}
\item If $b\in \mathcal{A}$ but $a\notin \mathcal{A}$, 
\begin{enumerate}
\item If $\mathcal{A}\cap[k-1]=q$, we have
\begin{align}
&D_{\text{KL}}(\mathbb{P}^a(V_{\mathcal{A}}^{(\ell)})||\mathbb{P}^b(V_{\mathcal{A}}^{(\ell)}))=\frac{q!(m-q)!}{m!}\mu(\delta)\nonumber\\
&\qquad\leq 2\omega_{m,q}p\delta^2\leq 2\omega_{m,q}p\frac{\delta^2}{1-\delta^2}.
\end{align}
\item If $\mathcal{A}\cap[k-1]=q-1$, we have
\begin{align}
&D_{\text{KL}}(\mathbb{P}^a(V_{\mathcal{A}}^{(\ell)})||\mathbb{P}^b(V_{\mathcal{A}}^{(\ell)}))=\frac{q!(m-q)!}{m!}\upsilon(\delta)\nonumber\\
&\qquad\qquad\leq 2\omega_{m,q}p\frac{\delta^2}{1-\delta^2}.
\end{align}
\item If $\mathcal{A}\cap[k-1]\neq q-1$ and $\mathcal{A}\cap[k-1]\neq q$, then $D_{\text{KL}}(\mathbb{P}^a(V_{\mathcal{A}}^{(\ell)})||\mathbb{P}^b(V_{\mathcal{A}}^{(\ell)}))=0$. 
\end{enumerate}
\item If $a,b\in \mathcal{A}$,
\begin{enumerate}
\item If $\mathcal{A}\cap[k-1]=q-1$, and $q>1$,
\begin{align}
&D_{\text{KL}}(\mathbb{P}^a(V_{\mathcal{A}}^{(\ell)})||\mathbb{P}^b(V_{\mathcal{A}}^{(\ell)}))=\frac{q!(m-q)!}{m!}[\mu(\delta)+\upsilon(\delta)]\nonumber\\
%&\qquad\qquad\qquad\qquad+q!(m-q)!(m!)^{-1}[\log (1+\delta)^{-1}+\log (1-\delta)^{-1}]\nonumber\\
&\qquad\leq 4\omega_{m,q}p\delta^2\leq 4\omega_{m,q}p\frac{\delta^2}{1-\delta^2}.
\end{align}
If $\mathcal{A}\cap[k-1]=q-1$, and $q=1$,
\begin{align}
&D_{\text{KL}}(\mathbb{P}^a(V_{\mathcal{A}}^{(\ell)})||\mathbb{P}^b(V_{\mathcal{A}}^{(\ell)}))\nonumber\\
&=\frac{(m-2)(m-2)!}{(m!)}[\mu(\delta)+\upsilon(\delta)]\nonumber\\
&\,\,+\frac{(m-2)!}{(m!)}p[(1-\delta)\log \frac{1-\delta}{1+\delta}+(1+\delta)\log\frac{1+\delta}{1-\delta}]\nonumber\\
%&\leq 2\frac{(m-2)(m-2)!}{m!}\delta^2+2\frac{(m-2)!}{m!}\frac{4\delta^2}{1-\delta^2}
&\leq 8\omega_{m,q}p\frac{\delta^2}{1-\delta^2}.
\end{align}
Thus for all $q$, we have
\begin{align}
&D_{\text{KL}}(\mathbb{P}^a(V_{\mathcal{A}}^{(\ell)})||\mathbb{P}^b(V_{\mathcal{A}}^{(\ell)}))\leq 8\omega_{m,q}p\frac{\delta^2}{1-\delta^2}.
\end{align}
\item $\mathcal{A}\cap[k-1]\neq q-1$, then $D_{\text{KL}}(\mathbb{P}^a(V_{\mathcal{A}}^{(\ell)})||\mathbb{P}^b(V_{\mathcal{A}}^{(\ell)}))=0$. 
\end{enumerate}
\end{enumerate}

Summarizing the above bound, we arrive at
\begin{align}
&D_{\text{KL}}(\mathbb{P}^a||\mathbb{P}^b)=r p\sum_{\mathcal{A}\subseteq [n]:|\mathcal{A}|=m} D_{\text{KL}}(\mathbb{P}^a(V_{\mathcal{A}}^{(1)})||\mathbb{P}^b(V_{\mathcal{A}}^{(1)}))\nonumber\\
&\leq r\frac{4\delta^2\omega_{m,q}p}{(1-\delta^2)}\left[{k-1 \choose q-1}{n-k-1 \choose m-q}\right.\nonumber\\
&\quad\left.+{k-1 \choose q}{n-k-1 \choose m-q-1}+2{k-1\choose q-1}{n-k-1 \choose m-q-1}\right]\nonumber\\
&\leq r\frac{4\delta^2\omega_{m,q}p}{(1-\delta^2)} \left[{k-1 \choose q-1}{n-k \choose m-q}+{k \choose q}{n-k-1 \choose m-q-1}\right]
\end{align}
which is the desired result. 
\end{proof}

\end{document}